\documentclass[journal]{IEEEtran}
\usepackage[T1]{fontenc}
\usepackage[latin9]{inputenc}
\usepackage{float}
\usepackage{bm}
\usepackage{amsmath}
\usepackage{amsthm}
\usepackage{amssymb}
\usepackage{graphicx}
\usepackage[unicode=true,
 bookmarks=true,bookmarksnumbered=true,bookmarksopen=true,bookmarksopenlevel=1,
 breaklinks=false,pdfborder={0 0 0},pdfborderstyle={},backref=false,colorlinks=false]
 {hyperref}
\hypersetup{pdftitle={Your Title},
 pdfauthor={Your Name},
 pdfpagelayout=OneColumn, pdfnewwindow=true, pdfstartview=XYZ, plainpages=false}

\makeatletter

\floatstyle{ruled}
\newfloat{algorithm}{tbp}{loa}
\providecommand{\algorithmname}{Algorithm}
\floatname{algorithm}{\protect\algorithmname}

\usepackage[caption=false,font=footnotesize]{subfig}
\usepackage{bm}

\usepackage{cite}
\usepackage[T1]{fontenc}
\usepackage{algorithm}
\usepackage{algorithmic}

\IEEEoverridecommandlockouts

\ifdefined\showcaptionsetup
 \PassOptionsToPackage{caption=false}{subfig}
\fi
\usepackage{subfig}
\makeatother

\theoremstyle{plain}
\newtheorem{prop}{\protect\propositionname}
\newtheorem{lem}{\protect\lemmaname}
\newtheorem{thm}{\protect\theoremname}
\theoremstyle{definition}
\newtheorem{defn}{\protect\definitionname}
\providecommand{\definitionname}{Definition}
\providecommand{\lemmaname}{Lemma}
\providecommand{\propositionname}{Proposition}
\providecommand{\theoremname}{Theorem}

\begin{document}
\title{Partially Observable Restless Bandits for Age-Optimal Scheduling over
Markov Channels}
\author{Xijun~Wang, Shuying Gan, Yanzhi~Huang, Xiaoyu Zhao, Chao Xu, and
Xiang~Chen\thanks{Part of this work was presented at the IEEE/CIC ICCC, Aug. 2020 \cite{9238787}.}\thanks{X. Wang, S. Gan, and X. Chen are with School of Electronics and Information
Technology, Sun Yat-sen University, Guangzhou, 510006, China (e-mail:
wangxijun@mail.sysu.edu.cn; ganshy7@mail2.sysu.edu.cn; chenxiang@mail.sysu.edu.cn).
}\thanks{Y. Huang is with Guangzhou Huizhi Communication Co.,Ltd., Guangzhou,
China. This work was done when he was with School of Electronics and
Communication Engineering, Sun Yat-sen University, Guangzhou, China
(e-mail: huangyz53@163.com).}\thanks{X. Zhao is with School of Cyber Science and Engineering Southeast
University Nanjing, Jiangsu, China (e-mail: xy-zhao@seu.edu.cn).}\thanks{C. Xu is with School of Information Engineering, Northwest A\&F University,
Yangling, Shaanxi, China (e-mail: cxu@nwafu.edu.cn). }}
\maketitle
\begin{abstract}
There is a surge of need for fresh information with the overwhelming
proliferation of the Internet of Things (IoT) applications. To characterize
the information freshness perceived by the destination, the age of
information (AoI) has been proposed. In this paper, we consider an
IoT system with multiple devices sending status update packets to
a central controller through time-correlated Markov channels and assume
that the instantaneous channel states are not available to the central
controller before making scheduling decisions. To ensure information
freshness, we investigate a timely scheduling problem that minimizes
the total expected time-average AoI under a strict communications
bandwidth constraint. We formulate this problem as a partially observable
restless multi-armed bandit problem. Using Lagrangian relaxation,
we decouple the relaxed problem into multiple sub-problems and prove
the threshold structure of their optimal policies. Armed with this
property, we establish the indexability for the decoupled problem
and design an algorithm to compute the Whittle's index.{}
To reduce implementation complexity, we further derive the Whittle-like
index in closed-form{} for low-complexity scheduling.
Simulation results show that the proposed index-based policies outperform
the baselines, remain close to the optimal policy or relaxed lower
bound, and are especially effective when scheduling resources are
limited or the network size is large.
\end{abstract}


\section{Introduction}

There is a rapidly emerging demand for time-critical applications
in the Internet of Things (IoT), such as industrial automation, health
monitoring, and smart surveillance. These applications necessitate
the timely transmission of status updates from various devices to
a central controller via wireless networks \cite{schulzLatencyCriticalIoT2017,wuCriticalInternetThings2020,jiang2020aiassisted}.
Within this context, a newly generated status update holds greater
value for the receiver than the outdated ones. Unfortunately, optimizing
the conventional performance metrics, such as delay and throughput,
does not inherently ensure the freshness of information. To address
this concern, the notion of the age of information (AoI) has been
introduced to capture the timeliness of the status update from the
perspective of the destination \cite{kaulRealtimeStatusHow2012}.
Essentially, the AoI measures the time elapsed since the generation
of the most recently received status update. Henceforth, the AoI becomes
a pivotal performance metric to enhance information freshness in IoT
systems \cite{xuOptimizingInformationFreshness2020,liuUAVAidedDataCollection2021,yangUnderstandingAgeInformation2021,xuOptimalStatusUpdate2021,wangWhenPreprocessKeeping2022,wangAgeChangedInformation2022}.

In an IoT system involving multiple devices, the bandwidth constraint
restricts the simultaneous scheduling of a limited number of devices
to transmit their status updates to the central controller. If a device
lacks transmission opportunities, its AoI will steadily increase.
Even if the device is scheduled to transmit, deep channel fading can
lead to transmission failure, resulting in an increase in AoI as well.
Therefore, how to effectively schedule IoT devices under the bandwidth
constraint to minimize the overall AoI is an overarching challenge.
Recent endeavors have been made to tackle this issue \cite{kadotaSchedulingPoliciesMinimizing2018,kadotaMinimizingAgeInformation2019,sunClosedFormWhittleIndexEnabled2020,hsuSchedulingAlgorithmsMinimizing2020,maatoukOptimalityWhittleIndex2020,SchedulingAlgorithmsforOptimizingAgeofInformationinWirelessNetworkswithThroughputConstraints,zhouMinimumAgeInformation2020a,ceranReinforcementLearningApproach2021a}.
These efforts primarily focus on either the reliable channel scenario,
where transmissions are consistently successful, or the erroneous
channel scenario, where transmission failures are independent and
identically distributed (i.i.d.) across time slots. However, many
practical scenarios involve time-correlated channel fading. Leveraging
this time correlation presents an opportunity to devise a more efficient
scheduling policy that enhances information freshness \cite{panMinimizingAgeInformation2021,yaoAgeOptimalLowPowerStatus2020}.
Nevertheless, these existing studies only consider a single transmitter
and receiver pair, thus leaving the problem of scheduling with multiple
devices unresolved.

In this paper, we consider a time-critical IoT system where multiple
devices generate status update packets and send them to a central
controller via time-correlated wireless channels. In particular, due
to the limited communication bandwidth, only a subset of devices can
be scheduled to send their update packets in each time slot, and hence
the AoI evolutions of the multiple devices intricately interact with
each other.  We model the wireless channel as a Markov chain with
two states, i.e., GOOD and BAD states. Specifically, the transmission
will succeed if the channel is in the GOOD state, and fail otherwise.
Moreover, we assume that the channel states of all the devices are
not known by the scheduler before the transmission. This is a common
and practical scenario in time-critical applications, as acquiring
this information would require additional control-plane functionalities,
leading to significant signaling overhead and increased delay. Consequently,
the channel states of the scheduled devices can only be inferred at
the end of each time slot based on the transmission outcome, while
the states of other devices remain unobservable. In this setup, we
aim to design a timely scheduling policy without the knowledge of
the instantaneous channel state to minimize the total expected time-average
AoI at the destination.{} The central scheduling difficulty
is to balance the urgency of stale information, captured by AoI, and
the transmission opportunity, captured by the channel belief under
partially observable time-correlated channels. In our earlier work
\cite{9238787}, we proposed a scheduling policy to minimize the AoI
under a special case of the Markov channel. This paper extends that
work by considering a more general channel model, providing a more
comprehensive framework for timely scheduling in such systems. The
key contributions of this paper are summarized as follows:

\begin{itemize}
\item We formulate the timely scheduling problem as a partially
observable Restless Multi-Armed Bandit (RMAB) problem since only the
channel states of the scheduled devices can be observed. We relax
the partially observable RMAB problem with the strict constraint and
decouple it into multiple sub-problems, each of which is cast into
a Partially Observable Markov Decision Process (POMDP). By introducing
a belief state, we convert the POMDP to a Markov Decision Process
(MDP), establish the threshold structure of the optimal policy for
each sub-problem, and further prove the indexability of the decoupled
problem.
\item We derive the Whittle's index in closed-form for
a special case and design an algorithm to compute the Whittle's index
in the general case. Based on the resulting indices, we propose a
Whittle-index-based scheduling policy that satisfies the strict bandwidth
constraint. To further reduce the computational complexity, we derive
a closed-form expression for a Whittle-like index and propose a corresponding
Whittle-like index based policy that schedules the devices with the
largest indices.
\item To evaluate the proposed policies, we design an algorithm
to solve the relaxed partially observable RMAB problem under an equality
constraint on the average number of scheduled devices, which provides
a lower bound to the original scheduling problem. We also use the
greedy and myopic policies as low-complexity benchmarks. Through extensive
simulations, we demonstrate that our Whittle's index and Whittle-like
index policies achieve significantly lower time-average AoI than the
baseline policies and remain close to the theoretical lower bound
as the number of devices increases. The results further show that
exploiting both AoI and channel-belief information is particularly
beneficial when scheduling resources are scarce or the network size
is large.
\end{itemize}
The remainder of the paper is organized as follows. Section II summarizes
related work, and Section III presents the system model. Section IV
formulates the partially observable RMAB problem, establishes the
indexability of the decoupled sub-problem, and presents the Whittle's
index policy and the Whittle-like index policy. Section V provides
the technical proofs of the main results. Simulation results are discussed
in Section VI, followed by the conclusion in Section VII.

\section{Related Work}

\subsection{Scheduling for AoI Minimization}

The minimization of average AoI in a single-hop network has been studied
under various settings \cite{kadotaSchedulingPoliciesMinimizing2018,kadotaMinimizingAgeInformation2019,sunClosedFormWhittleIndexEnabled2020,hsuSchedulingAlgorithmsMinimizing2020,maatoukOptimalityWhittleIndex2020,SchedulingAlgorithmsforOptimizingAgeofInformationinWirelessNetworkswithThroughputConstraints,zhouMinimumAgeInformation2020a,ceranReinforcementLearningApproach2021a}.
In \cite{kadotaSchedulingPoliciesMinimizing2018}, the authors considered
periodical packet generations and proposed several low-complexity
scheduling policies to minimize the AoI. With stochastic packet arrivals,
several scheduling policies were proposed in \cite{kadotaMinimizingAgeInformation2019}
to minimize the AoI under three queueing disciplines, i.e., FIFO queue,
one-buffer queue, and no queue. With the same packet arrival model,
a Whittle\textquoteright s index-based random access scheme was proposed
for status updating with a one-buffer queue in \cite{sunClosedFormWhittleIndexEnabled2020}.
In \cite{hsuSchedulingAlgorithmsMinimizing2020}, the age-optimal
scheduling problem was studied in the wireless broadcast network without
buffers at the source.

The above works focus on the case that the arrival of status update
packets follows a certain rule. There is another line of work studying
the age-optimal device scheduling, where the device generates a status
update packet immediately when it is scheduled to transmit. In \cite{maatoukOptimalityWhittleIndex2020},
the authors considered a wireless network where the users generate
status updates at will and send them to a central controller over
i.i.d. unreliable channels. Under the bandwidth constraint, the authors
developed a scheduling policy based on the Whittle\textquoteright s
index approach and established its asymptotic optimality. With the
same channel model, the authors in \cite{SchedulingAlgorithmsforOptimizingAgeofInformationinWirelessNetworkswithThroughputConstraints}
developed several scheduling policies to minimize the average AoI
under the throughput constraint of each user. By considering the non-uniform
packet size of the status updates, the authors in \cite{zhouMinimumAgeInformation2020a}
studied status sampling and device scheduling over unreliable channels.
The transmission scheduling of status updates from a source to multiple
users was investigated in \cite{ceranReinforcementLearningApproach2021a}
by considering both the standard ARQ and the hybrid ARQ protocols.
However, the time correlation in the channel fading process is ignored
in the above works.

\subsection{Markov Channel }

In many practical scenarios, the wireless channel fading exhibits
the characteristics of time correlation, which can be modeled as a
finite-state Markov chain. In \cite{tangMinimizingAgeInformation2020},
the authors considered a finite-state ergodic Markov channel and assumed
that the Channel State Information (CSI) is known to the central controller
before making scheduling decisions. As such, the transmit power can
be controlled in different channel states to guarantee successful
transmission. A scheduling policy was then proposed to minimize the
average AoI under the power consumption constraint. A simplified but
often-used channel model is a two-state Markov chain, also known as
the Gilbert\textendash Elliott channel, where the channel could be
either in a GOOD or BAD state. An energy harvesting transmitter communicating
over a Gilbert\textendash Elliott channel was considered in \cite{salehiheydarabadChannelSensingCommunication2018},
where the transmitter can sense the current channel state before the
transmission at the cost of both an increase in energy consumption
and a reduction in transmission time. The optimal transmission policy
was then designed to maximize the long-term throughput. Under the
same channel model, the status update system with a single device
was studied in \cite{panMinimizingAgeInformation2021,yaoAgeOptimalLowPowerStatus2020}
to minimize the average AoI, where the delayed CSI is always available
regardless of transmission decisions. A more practical case was also
considered in \cite{yaoAgeOptimalLowPowerStatus2020}, where the channel
state is not available to the transmitter when making a decision and
is revealed only after transmission occurs. The transmission scheduling
problem was hence formulated as a POMDP. However, the authors in \cite{yaoAgeOptimalLowPowerStatus2020}
only considered the transmission scheduling of a single user. In \cite{5605371},
multichannel access in cognitive radio systems was considered, where
the occupation of each channel was modeled as a two-state Markov chain,
and the channel state was not known by the transmitter. This problem
was formulated as a partially observable RMAB and a Whittle's index
policy was proposed. Nonetheless, the throughput was maximized, other
than the AoI which evolves with time.

\subsection{Restless Multi-Armed Bandit}

RMAB is a generalization of Multi-Armed Bandit (MAB), where the states
of arms are frozen when they are not pulled. In contrast, the states
of arms in RMAB continuously evolve even if the arms are not pulled.
Gittin's index policy was proposed to solve the MAB problem \cite{GittinsBandit}.
Similarly, Whittle provided a heuristic index policy, called Whittle's
index policy, to the relaxation of the RMAB problem \cite{whittleRestlessBanditsActivity1988}.
Whittle's index policy is asymptotically optimal when the underlying
MDP satisfies indexability. It is also shown that Whittle's index
policy performs empirically well in general, making it a popular method
for RMABs. In classical RMAB problems, the states of all the arms
are completely observable in each time slot, regardless of whether
the arms are pulled or not \cite{whittleRestlessBanditsActivity1988,Restlessbanditspartialconservationlawsandindexability,Nino-Mora:2007aa}.
This model is further generalized to the partially observable RMAB
where only the states of arms that are pulled in the current slot
can be observed \cite{5605371,AsymptoticallyoptimaldownlinkschedulingoverMarkovianfadingchannels}.
In this work, the channel state is revealed to the central controller
only when the corresponding device is scheduled, therefore the timely
scheduling over the Markov channel can be formulated as a partially
observable RMAB problem.

\section{System Model}

We consider an IoT system consisting of $M$ IoT devices and a destination
(e.g., a central controller). The IoT devices sample the underlying
physical processes and send the status update packets to the destination
over time-correlated wireless channels. Time is considered to be slotted
and the duration of each time slot is normalized to unity. In each
time slot, a subset of the IoT devices is scheduled to update their
status information to the destination. For each IoT device, we let
$u_{i}(t)\in\{0,1$\} denote the scheduling decision in slot $t$,
where $u_{i}(t)=1$ indicates that the device $i$ is scheduled to
transmit the update to the destination, and $u_{i}(t)=0$, otherwise.
The scheduling decision of the system in slot $t$ is then denoted
by $\bm{U}(t)=[u_{1}(t),u_{2}(t),...,u_{M}(t)]$. We assume that each
IoT device can generate status updates at will. In other words, if
the IoT device is scheduled to update its status, it will generate
a fresh status update at the beginning of the slot. We also assume
that each IoT device takes one time slot to transmit a status update
packet to the destination. Due to the scarcity of the available bandwidth,
we suppose that at most $K$ ($K<M$) users can
transmit simultaneously in a time slot. Therefore, we have the constraint
that $\sum^{M}_{i=1}u_{i}(t)\leq K$ for any $t$.

The wireless channel between each IoT device and the destination is
modeled as a Gilbert\textendash Elliott channel as shown in Fig. \ref{fig:ChannelModel}.
The channel state is assumed to be independent across different channels.
Let $h_{i}(t)\in\{0,1\}$ denote the state of the channel between
device $i$ and the destination in slot $t$. Particularly, if the
channel of device $i$ is in the GOOD state in slot $t$, denoted
by $h_{i}(t)=1$, the status update packet will be successfully received
by the destination; while if the channel of device $i$ is in the
BAD state in slot $t$, denoted by $h_{i}(t)=0$, the transmission
will be failed. If the status update of device $i$ is not successfully
received, it will be dropped and a new status update will be generated
when device $i$ is scheduled the next time. We assume that the duration
of the channel state is equal to the slot length and the transition
occurs at the beginning of each slot.\footnote{In practical IoT implementations, the slot duration
should be chosen no longer than the channel coherence time, or sufficiently
small so that the channel state can be regarded as approximately constant
during packet transmission and feedback.} The transition probabilities are given by $\Pr[h_{i}(t+1)=1|h_{i}(t)=1]=\alpha_{i}$
and $\Pr[h_{i}(t+1)=0|h_{i}(t)=0]=\beta_{i}$, where $0<\alpha_{i}<1$
and $0<\beta_{i}<1$.{} The channel transition probabilities
are assumed to be stationary over the considered time horizon and
available to the scheduler, which can be obtained from historical
measurements or offline estimation. Similar to \cite{yaoAgeOptimalLowPowerStatus2020}
and \cite{salehiheydarabadChannelSensingCommunication2018}, we assume
positively correlated channels with $\alpha_{i}\geq1-\beta_{i}$ for
any IoT devices. We consider the case that the channel states in each
time slot $[h_{1}(t),h_{2}(t),\ldots,h_{M}(t)]\in\{0,1\}^{M}$ are
not known by the destination before the scheduling decision is made.
However, the channel states of a subset of devices can be observed
at the end of each time slot. Let $o_{i}(t)$ denote the observation
of channel state of device $i$ in slot $t$. For each device that
is scheduled to transmit in slot $t$, the destination observes a
successful transmission, denoted by $o_{i}(t)=1$, if the channel
is in the GOOD state; otherwise, the destination observes a transmission
failure, denoted by $o_{i}(t)=0$. For other devices that are not
scheduled in slot $t$, the observations at the destination are set
to be \emph{none}, i.e., $o_{j}(t)=\text{\emph{none}}$.
\begin{figure}[t]
\centering

\subfloat[Illustration of a typical Internet of Things (IoT) network.]{\centering

\includegraphics[width=0.9\columnwidth]{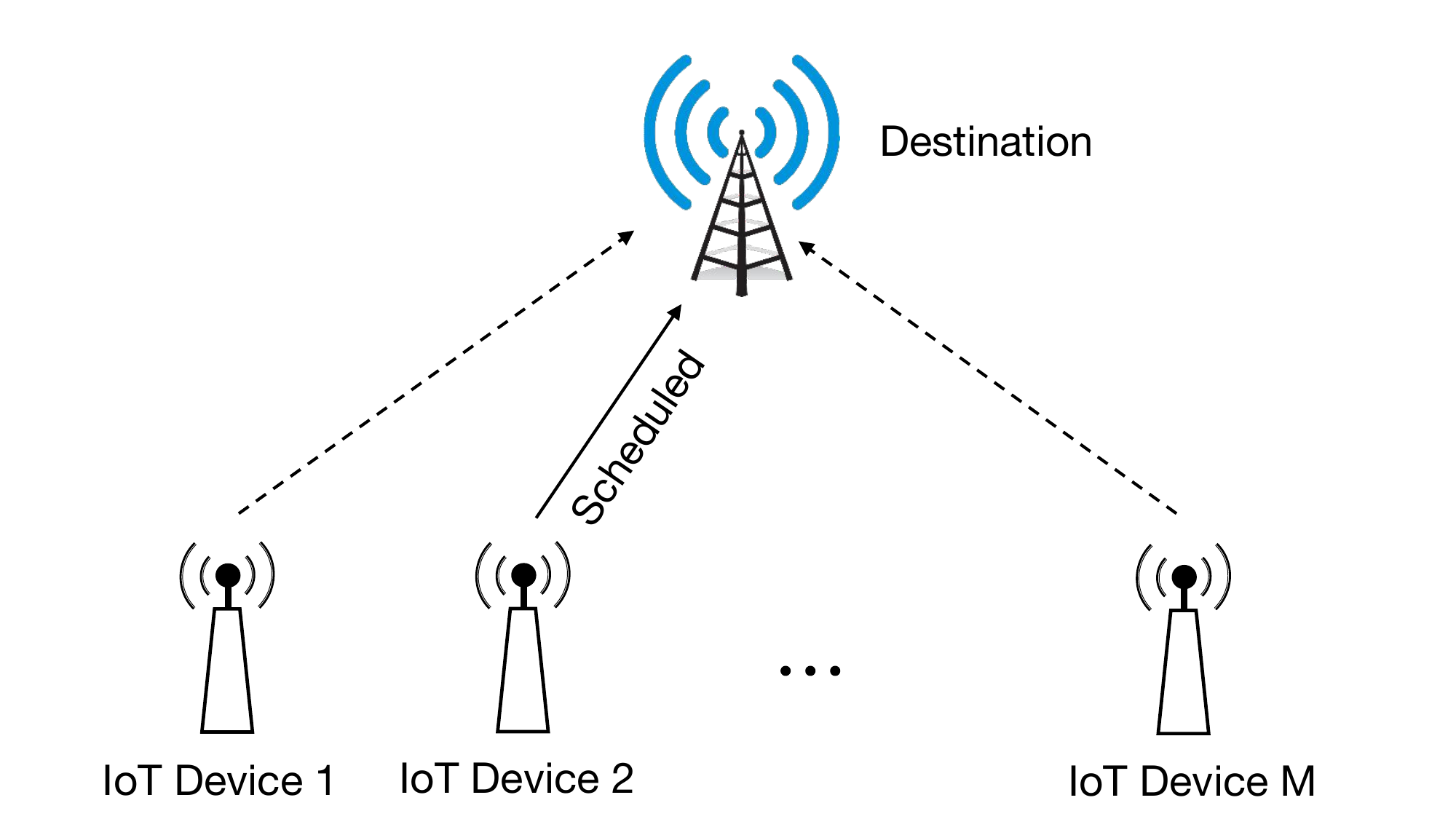}

}

\subfloat[\label{fig:ChannelModel}The Gilbert\textendash Elliott channel model.]{\centering

\includegraphics[width=0.35\textwidth]{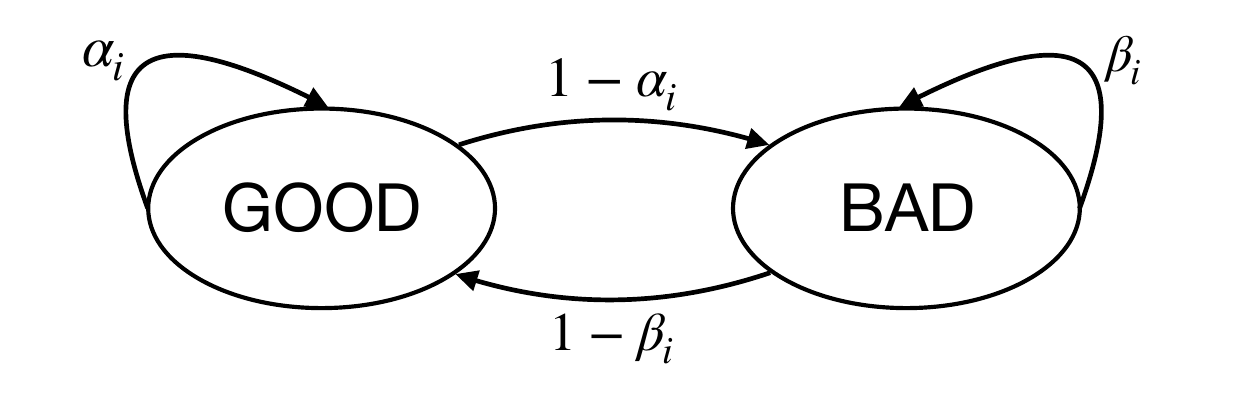}

}

\caption{Illustration of system model.}
\end{figure}

The freshness of the information is measured by the AoI, which is
defined as the time elapsed since the generation of the most recent
status update received by the destination. Let $\delta_{i}(t)$ be
the AoI of device $i$ at the beginning of slot $t$ and $\bm{\Delta}(t)=[\delta_{1}(t),\delta_{2}(t),...,\delta_{M}(t)]$
be the system AoI in slot $t$. If the update of the scheduled device
$i$ is successfully received by the destination, the AoI $\delta_{i}(t)$
will decrease to 1 (i.e., the transmission time of the delivered update);
otherwise the AoI will increase by 1. Therefore, the dynamics of AoI
$\delta_{i}(t)$ are expressed as follows:
\begin{equation}
\delta_{i}(t+1)=\begin{cases}
1, & u_{i}(t)=1,h_{i}(t)=1,\\
\delta_{i}(t)+1, & u_{i}(t)=1,h_{i}(t)=0,\\
\delta_{i}(t)+1, & u_{i}(t)=0.
\end{cases}\label{eq:AoI-update}
\end{equation}

A scheduling policy $\pi$ is defined as a sequence of actions, i.e.,
$\pi=[U(1),U(2),\ldots]$. We aim at designing the scheduling policy
so as to minimize the total expected time-average AoI at the destination.
By letting $\Pi$ be the set of all causal scheduling policies, i.e.,
scheduling decisions are made without future information, our scheduling
problem can be formulated as follows:
\begin{align}
\min_{\pi\in\Pi} & \limsup_{T\rightarrow\infty}\frac{1}{T}\mathbb{E}_{\pi}\left[\sum^{M}_{i=1}\sum^{T}_{t=1}\delta^{\pi}_{i}(t)\right],\\
\text{s.t.} & \enskip\sum^{M}_{i=1}u_{i}(t)\leq K,\forall t,
\end{align}
where $T$ is the length of time horizon and $\delta^{\pi}_{i}(t)$
is the AoI of device $i$ in slot $t$ under policy $\pi$.

It is noteworthy that each device in this problem can be considered
a restless bandit since the channel states and the AoI of all devices
evolve no matter whether the device is scheduled or not. However,
the channel state cannot be observed by the destination unless the
associated device is scheduled. Therefore, the above scheduling problem
can be formulated as a partially observable RMAB problem, which is
in general PSPACE-hard and difficult to find the optimal solution
\cite{doi:10.1287/moor.24.2.293}. To circumvent this difficulty,
we employ the Whittle\textquoteright s index policy, which has low
computational complexity and near-optimal performance \cite{whittleRestlessBanditsActivity1988}.
We then present the Whittle's index policy in detail in the following
section.

\section{Whittle's Index Policy}

In this section, we first relax the strict constraint in each time
slot and decouple the relaxed problem into multiple sub-problems via
a Lagrangian approach. Each sub-problem is formulated as a POMDP problem
and then reformulated as an MDP problem by introducing a belief state.
We show that the optimal policy for each sub-problem has an appealing
threshold structure, and prove the indexability property, which confirms
the existence of the Whittle's index. After that, we propose an algorithm
to compute the Whittle's index, based on which the scheduling policy
is presented.{} To further reduce the computational
complexity, we also derive a closed-form Whittle-like index and present
the corresponding Whittle-like index policy. We also study the optimal
policy for the relaxed partially observable RMAB problem, which serves
as a performance lower bound, and design a low-complexity myopic policy
as a benchmark. 

\subsection{Decoupled Sub-problem}

In the original scheduling problem, at most $K$ devices are chosen
in each time slot. We first relax this instantaneous constraint to
a time-average constraint, and obtain the relaxed problem as follows:
\begin{align}
\min_{\pi\in\Pi} & \limsup_{T\rightarrow\infty}\frac{1}{T}\mathbb{E}_{\pi}\left[\sum^{M}_{i=1}\sum^{T}_{t=1}\delta^{\pi}_{i}(t)\right],\\
\text{s.t.} & \enskip\frac{1}{T}\sum^{M}_{i=1}\sum^{T}_{t=0}\mathbb{E}_{\pi}[u_{i}(t)]\leq K.
\end{align}
Next, we employ a Lagrangian approach to transform the relaxed problem
into an unconstrained one. Let $W$ denote the Lagrangian parameter.
For any given $W\geq0$, the Lagrangian function of the relaxed problem
is given by 
\begin{equation}
\limsup_{T\rightarrow\infty}\frac{1}{T}\mathbb{E}_{\pi}\left[\sum^{M}_{i=1}\sum^{T}_{t=1}\left(\delta^{\pi}_{i}(t)+Wu_{i}(t)\right)\right]-WK.
\end{equation}
Since the last term $WK$ is irrespective of the policy $\pi$, the
dual problem can be expressed as 
\begin{equation}
\min_{\pi\in\Pi}\limsup_{T\rightarrow\infty}\frac{1}{T}\mathbb{E}_{\pi}\left[\sum^{M}_{i=1}\sum^{T}_{t=1}\left(\delta^{\pi}_{i}(t)+Wu_{i}(t)\right)\right].
\end{equation}
Then, we decompose the dual problem into $M$ sub-problems which can
be solved independently. By dropping the device index, the sub-problem
can be formulated as follows:
\begin{equation}
\min_{\pi\text{\ensuremath{\in}}\Pi}\limsup_{T\rightarrow\infty}\frac{1}{T}\mathbb{E}_{\pi}\left[\sum^{T}_{t=1}(\delta^{\pi}(t)+Wu(t))\right],
\end{equation}
where $W$ can be regarded as an additional cost to schedule the device.

In each sub-problem, an action is chosen at the beginning of each
time slot. If the IoT device is scheduled to transmit the status update,
i.e., $u(t)=1$, the destination is able to observe the state of the
channel. Otherwise, the channel state cannot be observed. No matter
what action is taken, the AoI will be updated according to (\ref{eq:AoI-update})
at the end of the time slot. Since the channel state is not directly
observable at the beginning of each time slot, each sub-problem can
be cast into a POMDP problem. To address this problem, we introduce
a belief state $\theta(t)$ which summaries the knowledge of the channel
state at the beginning of slot $t$ based on all past actions and
observations. In particular, the belief state is defined as the conditional
probability that the channel is in a GOOD state at the beginning of
a slot given the action and observation history. With action $u(t)$
and observation $o(t)$ at time slot $t$, the value of the belief
state is updated as follows:
\begin{equation}
\theta(t+1)=\begin{cases}
\alpha, & u(t)=1,o(t)=1,\\
1-\beta, & u(t)=1,o(t)=0,\\
\mathcal{T}(\theta(t)), & u(t)=0,
\end{cases}\label{eq:belief-update}
\end{equation}
where $\mathcal{T}(\theta(t))\triangleq\theta(t)\alpha+(1-\theta(t))(1-\beta)$
is the one-step belief state evolution when the channel state cannot
be observed. We also denote by $\mathcal{T}^{k}(\theta(t))\triangleq\mathcal{T}(\mathcal{T}^{k-1}(\theta(t)))$
the $k$-step belief state evolution when the channel state is unobserved
for $k$ consecutive slots. Additionally, we let $\mathcal{T}^{0}(\theta)=\theta$.
It has been shown in  \cite{10.2307/168926} that the belief state
is a sufficient statistic to describe the knowledge of the underlying
channel state for any slot $t$. By introducing the belief state,
we can now reformulate the POMDP as an infinite horizon average cost
MDP, whose components are described in the following.

\textbf{States}: The state of the MDP is defined as a tuple $\bm{s}(t)\triangleq(\delta(t),\theta(t))$,
where $\delta(t)\in\mathcal{A}\triangleq\{1,2,\ldots\}$ is the AoI
and $\theta(t)$ is the belief. Since the belief is either $\alpha$
or $1-\beta$ after every transmission, it evolves as $\mathcal{T}^{k}(\alpha)$
or $\mathcal{T}^{k}(1-\beta)$ before the next transmission, where
$0\leq k\leq\delta-1$. Hence, given the AoI $\delta$, the belief
state belongs to a finite set  $\mathcal{B}_{\delta}\triangleq\{\theta:\theta=\mathcal{T}^{\delta-1}(\alpha)\text{ or }\mathcal{T}^{k}(1-\beta),0\leq k\leq\delta-2\}$.
Altogether, the state space, which is given by $\mathcal{S}\triangleq\{(\delta,\theta):\delta\in\mathcal{A},\theta\in\mathcal{B}_{\delta}\}$,
is countably infinite. In the following, we will use $\bm{s}$ and
$(\delta,\theta)$ to represent the state interchangeably.

\textbf{Actions}: The action in slot $t$ is $u(t)$, which specifies
whether the IoT device is scheduled to update or not. The action space
is finite.

\textbf{Transition probabilities}: We let $\Pr[\bm{s}(t+1)\mid\bm{s}(t),u(t)]$
denote the probability that state transits from $\bm{s}(t)$ to $\bm{s}(t+1)$
in the next slot by taking action $u(t)$ in slot $t$. Specifically,
given the current state $\bm{s}(t)=(\delta,\theta)$, we have 
\begin{equation}
\left\{ \begin{aligned}\Pr[\bm{s}(t+1) & =(1,\alpha)\mid\bm{s}(t),u(t)=1]=\theta,\\
\Pr[\bm{s}(t+1) & =(\delta+1,1-\beta)\mid\bm{s}(t),u(t)=1]=1-\theta,\\
\Pr[\bm{s}(t+1) & =(\delta+1,\mathcal{T}(\theta))\mid\bm{s}(t),u(t)=0]=1,
\end{aligned}
\right.\label{eqMDPTrans}
\end{equation}
and $\Pr[\bm{s}(t+1)\mid\bm{s}(t),u(t)]=0$ otherwise. The dynamics
of AoI and belief are shown in Fig. \ref{fig:a_simple_path_of_AoI=000026belief}.
\begin{figure}[t]
\vspace{-2em}

\centering

\includegraphics[width=0.5\textwidth]{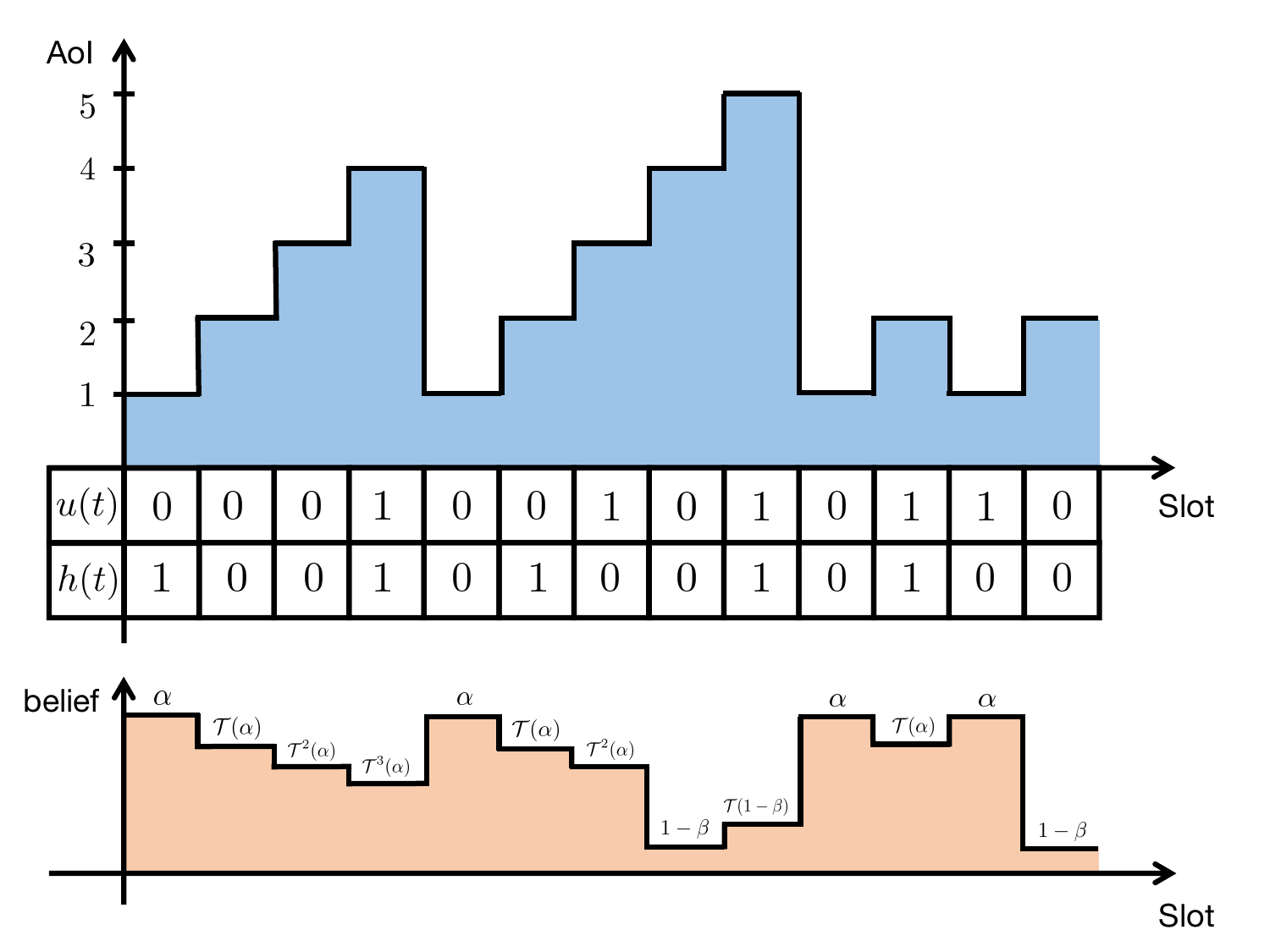}\vspace{-1em}

\caption{An illustration of the dynamics of AoI and belief in a time-slotted
IoT system.\label{fig:a_simple_path_of_AoI=000026belief}}
\vspace{-1.5em}
\end{figure}

\textbf{Cost}: Let $C(\bm{s}(t),u(t))$ be the immediate cost at state
$\bm{s}(t)$ with action $u(t)$, which consists of the AoI in the
current time slot and the extra cost of being scheduled. Specifically,
\begin{equation}
C(\bm{s}(t),u(t))=\delta(t)+Wu(t).
\end{equation}

Given an initial state $\bm{s}$, the average cost under a scheduling
policy $\pi$ can be expressed as:
\begin{equation}
V_{\pi}(\bm{s})=\limsup_{T\rightarrow\infty}\frac{1}{T}\mathbb{E}_{\pi}\left[\sum^{T}_{t=1}C(\bm{s}(t),u(t))\mid\bm{s}\right].
\end{equation}
The objective is to find a scheduling policy $\pi\in\Pi$ that minimizes
the average cost. In particular, a scheduling policy is average cost
optimal if it solves the following problem
\begin{equation}
\min_{\pi\in\Pi}V_{\pi}(\bm{s}),\label{eq:MDP}
\end{equation}
and the optimal scheduling policy is denoted by $\pi^{*}$. 

\subsection{Structure of Average Cost Optimal Policy}

A policy of the MDP is said to be stationary deterministic if it is
time-invariant and makes decisions with certainty. In this work,
we target at studying stationary deterministic policies which are
the simplest to be implemented. However, since the formulated infinite
horizon average cost MDP has a countably infinite state space and
unbounded costs, a stationary deterministic average cost optimal policy
may not exist. Therefore before we proceed further, we prove the existence
of a stationary deterministic average cost optimal policy for the
formulated MDP. To this end, we begin with presenting an infinite
horizon discounted cost MDP, and then we connect it to the average
cost MDP. Given an initial state $\bm{s}$, the expected total $\mu$-discounted
cost under a scheduling policy $\pi$ is given by
\begin{equation}
V_{\mu,\pi}(\bm{s})=\limsup_{T\rightarrow\infty}\mathbb{E}_{\pi}\left[\sum^{T}_{t=1}\mu^{t-1}C(\bm{s}(t),u(t))\mid\bm{s}\right],
\end{equation}
where $\mu\in(0,1)$ is a discount factor. A scheduling policy is
$\mu$-discounted cost optimal if it minimizes the expected total
$\mu$-discounted cost $V_{\mu,\pi}(\bm{s})$. Let $V_{\mu}(\bm{s})=\min_{\pi}V_{\mu,\pi}(\bm{s})$
denote the optimal expected total $\mu$-discounted cost. It is easy
to verify that there exists a stationary deterministic policy $f$
(i.e., a stationary deterministic policy that chooses $u(t)=0$ for
all time slots) such that $V_{\mu,f}(\bm{s})<\infty$ for every $\bm{s}$
and $\mu$. Then, according to \cite{sennottAverageCostOptimal1989},
we have the discounted cost optimality equation of $V_{\mu}(\bm{s})$
in the following proposition.
\begin{prop}
\label{prop:discounted cost}(a) For an initial state $\bm{s}$, the
optimal expected total $\mu$-discounted cost $V_{\mu}(\bm{s})$ satisfies
the following discounted cost optimality equation:
\begin{equation}
V_{\mu}(\bm{s})=\min_{u\in\{0,1\}}Q_{\mu}(\bm{s},u),\label{eq:discount Bellman function}
\end{equation}
where $Q_{\mu}(\bm{s},u)=C(\bm{s},u)+\mu\sum_{\bm{s}'\in\mathcal{S}}\Pr\left[\bm{s}'\mid\bm{s},u\right]V_{\mu}(\bm{s}')$
is the state-action value function (i.e., Q-function). The stationary
deterministic policy determined by the right-hand-side of the discounted
cost optimality equation in Eq. (\ref{eq:discount Bellman function})
is $\mu$-discounted cost optimal.

(b) The optimal expected total $\mu$-discounted cost $V_{\mu}(\bm{s})$
can be obtained through a value iteration algorithm. Let $V_{\mu,0}(\bm{s})=0$
for all $\bm{s}\in\mathcal{S}$ and for any $n\geq0$,
\begin{equation}
V_{\mu,n+1}(\bm{s})=\min_{u\in\{0,1\}}Q_{\mu,n+1}(\bm{s},u),\label{eq:Bellman in VIA}
\end{equation}
where
\begin{equation}
Q_{\mu,n+1}(\bm{s},u)=C(\bm{s},u)+\mu\sum_{\bm{s}'\in\mathcal{S}}\Pr\left[\bm{s}'\mid\bm{s},u\right]V_{\mu,n}(\bm{s}').
\end{equation}
Then, we have $V_{\mu,n}(\bm{s})\rightarrow V_{\mu}(\bm{s})$ as $n\rightarrow\infty$,
for every $\bm{s}$ and $\mu$.
\end{prop}

By using the value iteration in Proposition \ref{prop:discounted cost},
we present properties of $V_{\mu}(\bm{s})$ in the following lemmas.
\begin{lem}
\label{lem:Lemma1:discount_monitonicity}For any given $\theta$,
$V_{\mu}(\delta,\theta)$ is a non-decreasing function in $\delta$.
\end{lem}
\begin{IEEEproof}
Let $\bm{s}_{1}=(\delta_{1},\theta)$ and
$\bm{s}_{2}=(\delta_{2},\theta)$ with $\delta_{1}\leq\delta_{2}$.
We prove by induction over the value-iteration index $n$ that $V_{\mu,n}(\bm{s}_{1})\leq V_{\mu,n}(\bm{s}_{2})$.
The claim holds for $n=0$ since $V_{\mu,0}(\bm{s})=0$. Suppose it
holds at iteration $n$. For action $u=0$, the Bellman update gives
\begin{align}
 & Q_{\mu,n+1}(\bm{s}_{1},0)\nonumber \\
 & =\delta_{1}+\mu V_{\mu,n}(\delta_{1}+1,\mathcal{T}(\theta))\nonumber \\
 & \leq\delta_{2}+\mu V_{\mu,n}(\delta_{2}+1,\mathcal{T}(\theta))=Q_{\mu,n+1}(\bm{s}_{2},0).
\end{align}
For action $u=1$, we similarly have
\begin{align}
 & Q_{\mu,n+1}(\bm{s}_{1},1)\nonumber \\
 & =\delta_{1}+W+\mu\theta V_{\mu,n}(1,\alpha)\nonumber \\
 & \quad+\mu(1-\theta)V_{\mu,n}(\delta_{1}+1,1-\beta)\nonumber \\
 & \leq\delta_{2}+W+\mu\theta V_{\mu,n}(1,\alpha)\nonumber \\
 & \quad+\mu(1-\theta)V_{\mu,n}(\delta_{2}+1,1-\beta)=Q_{\mu,n+1}(\bm{s}_{2},1).
\end{align}
Taking the minimum over $u\in\{0,1\}$ preserves the order, so $V_{\mu,n+1}(\bm{s}_{1})\leq V_{\mu,n+1}(\bm{s}_{2})$.
Letting $n\rightarrow\infty$ proves that $V_{\mu}(\delta,\theta)$
is non-decreasing in $\delta$.
\end{IEEEproof}
\begin{lem}
\label{lem:lemma2:concavity} For any given $\delta$, $V_{\mu}(\delta,\theta)$
is a non-increasing and concave function with respect to $\theta$.
\end{lem}
\begin{IEEEproof}
The proof uses the same value-iteration induction
as Lemma \ref{lem:Lemma1:discount_monitonicity}, but applies it to
the belief dimension. The key steps are to show that the Q-function
of idle action preserves monotonicity through $\mathcal{T}(\theta)$,
while the Q-function of active action is affine in $\theta$. These
two properties establish both non-increasingness and concavity of
$V_{\mu}(\delta,\theta)$ with respect to $\theta$. The detailed
proof is provided in Section \ref{subsec:Proof-of-Lemma2}.
\end{IEEEproof}
Based on the properties of $V_{\mu}(\bm{s})$ in Lemmas \ref{lem:Lemma1:discount_monitonicity}
and \ref{lem:lemma2:concavity}, we show that the MDP in (\ref{eq:MDP})
has a stationary deterministic optimal policy in the following theorem.
\begin{thm}
\label{thm:existence_optimal}There exists a stationary deterministic
policy that is average cost optimal. The optimal average cost is $V^{*}=\lim_{\mu\rightarrow1}(1-\mu)V_{\mu}(\bm{s})$
and independent of the initial state $\bm{s}$.
\end{thm}
\begin{IEEEproof}
The proof verifies sufficient conditions for
average-cost optimality. Specifically, we construct a stationary deterministic
policy with finite average cost and then use Lemmas \ref{lem:Lemma1:discount_monitonicity}
and \ref{lem:lemma2:concavity} to lower-bound the relative discounted
value function. The detailed proof is provided in Section \ref{subsec:Proof-of-Theorem1}.
\end{IEEEproof}
Next, we investigate the structural property of the average cost optimal
policy, which facilitates the proof of indexability in the next subsection.
Unfortunately, although we have shown the optimality of a stationary
deterministic policy, the average cost optimality equation for the
MDP with a countably infinite state space and unbounded cost might
not be available \cite{hsuSchedulingAlgorithmsMinimizing2020}, \cite{CAVAZOSCADENA1991387}.
To tackle this challenge, in the following theorem, we show the structure
of $\mu$-discounted cost optimal policy via the discounted cost optimality
equation, and then extend it to the average cost optimal policy.
\begin{thm}
\label{thm:threshold}The average cost optimal policy is of threshold-type
in $\theta$, that is, for any $\delta$ and $\theta_{1}\leq\theta_{2}$,
if it is optimal to take action $u=1$ at state $(\delta,\theta_{1}),$
then it is also optimal to take action $u=1$ at state $(\delta,\theta_{2})$.
\end{thm}
\begin{IEEEproof}
The proof first establishes the threshold
structure for the discounted-cost problem through value iteration.
The induction compares the Q-functions of the active and idle actions
as functions of the belief state and uses Lemma \ref{lem:lemma2:concavity}
to show that the action-switching point is a belief threshold. The
average-cost result then follows by a vanishing-discount argument
using the existence result in Theorem \ref{thm:existence_optimal}.
The detailed proof is provided in Section \ref{subsec:Proof-of-Theorem2}.
\end{IEEEproof}

\subsection{Whittle's Index Policy}

The premise of applying Whittle's index policy is to prove the indexability
of the sub-problem, which may not be satisfied by all RMAB problems
and is usually difficult to establish. In the previous subsection,
we have shown the threshold structure of the average cost optimal
policy, based on which we establish the indexability for the decoupled
problem. We first give the definition of indexability.
\begin{defn}
\label{def:indexable}The sub-problem is indexable if a passive set,
i.e., the set of states for which it is optimal to make the device
idle, increases monotonically from the empty set to the entire state
space with the increasing of the extra cost $W.$
\end{defn}
\begin{thm}
\label{thm:indexability}The sub-problem defined in (\ref{eq:MDP})
is indexable.
\end{thm}
\begin{IEEEproof}
The proof uses the threshold structure in
Theorem \ref{thm:threshold}. For each fixed AoI, increasing the activation
charge $W$ shifts the belief threshold monotonically upward, so the
passive set expands from the empty set to the entire belief space.
This monotone expansion is exactly the indexability condition. The
detailed proof is provided in Section \ref{subsec:Proof-of-Theorem8}.
\end{IEEEproof}
After establishing the indexability of the sub-problem, we can calculate
the Whittle's index, which is the infimum cost $W$ that makes both
actions equally desirable, for each state. We first consider a special
case, where $\alpha=1-\beta$. In this case, the belief is irrespective
of time, i.e., $\theta(t)=\alpha$ for all $t$, and the channel realization
is independent and identically distributed over time. Hence, the state
is reduce to the AoI, i.e., $s(t)=\delta(t)$. In the following theorem,
we obtain the Whittle's index in closed-form.
\begin{thm}
\label{thm9:special_case}When $\alpha=1-\beta$, the Whittle's index
for state $\delta$ is
\begin{equation}
W_{\delta}=\delta+\frac{\delta(\delta-1)}{2}\alpha.
\end{equation}
\end{thm}
\begin{IEEEproof}
When $\alpha=1-\beta$, the belief state is
constant and the single-arm MDP reduces to an AoI-only threshold problem.
We compute the stationary distribution of the DTMC induced by a threshold
$n$, obtain the corresponding average cost, and identify the Whittle
index as the value of $W$ at which the average costs under thresholds
$\delta$ and $\delta+1$ are equal. The detailed proof is provided
in Section \ref{subsec:Proof-of-Theorem9}.
\end{IEEEproof}
Next, we investigate the Whittle's index in the general case, where
$\alpha\neq1-\beta$. Unfortunately, obtaining a closed-form expression
for the Whittle's index is very difficult, if not impossible. Hence,
we propose an algorithm to compute the Whittle's index by calculating
the Q-function. Since the state of the MDP is countably infinite,
it is impractical to directly apply a classic relative value iteration
(RVI) algorithm \cite{dimitrip.bertsekasDynamicProgrammingOptimal2007}
to obtain the Q-function. Therefore, we construct a finite-state MDP
to approximate the original MDP in (\ref{eq:MDP}) and show the convergence
of this finite-state approximation.

First, we construct the finite-state MDP by truncating the value of
the AoI and the number of Markov transitions. In particular, the AoI
is upper-bounded by $\hat{\delta}$, which can be arbitrary large,
and the corresponding belief state is limited by $\mathcal{T}^{\hat{\delta}-1}(\alpha)$
or $\mathcal{T}^{\hat{\delta}-2}\left(1-\beta\right)$. We assume
$\hat{\delta}>M$ without loss of generality. The state space of the
approximate MDP is then given by $\mathcal{S}^{\hat{\delta}}\triangleq\{(\delta,\theta):\delta\in\{1,2,\ldots,\hat{\delta}\},\theta\in\mathcal{B}_{\delta}\}$.
Since the Markov channel is positively correlated, we have $\mathcal{T}^{k-1}(\alpha)>\mathcal{T}^{k}(\alpha)$
and $\mathcal{T}^{k-1}(1-\beta)<\mathcal{T}^{k}(1-\beta)$. Moreover,
as $k\rightarrow\infty$, $\mathcal{T}^{k}(\alpha)$ decreasingly
converges to $\theta^{*}$, while $\mathcal{T}^{k}(1-\beta)$ increasingly
converges to $\theta^{*}$, where $\theta^{*}=\frac{1-\beta}{2-\alpha-\beta}$
is the equilibrium belief state. Therefore, we have $\mathcal{T}^{\hat{\delta}-1}(\alpha)>\mathcal{T}^{\hat{\delta}-1}(1-\beta)>\mathcal{T}^{\hat{\delta}-2}(1-\beta)$.
The evolutions of the bounded AoI and belief are given, respectively,
as follows:
\begin{equation}
\delta(t+1)=\begin{cases}
1, & u(t)=1,h(t)=1,\\
\min\{\hat{\delta},\delta(t)+1\}, & u(t)=1,h(t)=0,\\
\min\{\hat{\delta},\delta(t)+1\}, & u(t)=0,
\end{cases}\label{eq:AoI-update-bound}
\end{equation}
and
\begin{equation}
\theta(t+1)=\begin{cases}
\alpha, & u(t)=1,o(t)=1,\\
1-\beta, & u(t)=1,o(t)=0,\\
\mathcal{T_{\hat{\delta}}}(\theta(t)), & u(t)=0,
\end{cases}\label{eq:belief-update-bound}
\end{equation}
where
\begin{equation}
\mathcal{T_{\hat{\delta}}}(\theta(t))=\begin{cases}
\mathcal{T}^{\hat{\delta}-1}(\alpha), & \mathcal{T}^{\hat{\delta}-2}(1-\beta)<\theta(t)<\mathcal{T}^{\hat{\delta}-1}(\alpha),\\
\mathcal{T}(\theta(t)), & \text{ otherwise}.
\end{cases}
\end{equation}

By increasing $\hat{\delta}$, we obtain a sequence of finite-state
approximate MDPs. Then, we show the convergence of the proposed sequence
of approximate MDPs to the original MDP in the following theorem.
\begin{thm}
\label{thm:Approximat}Let $V^{(\hat{\delta})*}$ be the minimum average
cost of the approximate MDP with respect to the upper bound $\hat{\delta}.$
We have $V^{(\hat{\delta})*}\rightarrow V^{*}$ as $\hat{\delta}\rightarrow\infty$.
\end{thm}
\begin{IEEEproof}
The proof verifies the convergence conditions
for finite-state approximations of countable-state average-cost MDPs.
We first bound the relative discounted value function of the truncated
MDP by a finite first-passage cost and then show that the optimal
average cost of the truncated MDP is asymptotically no larger than
that of the original MDP. These two steps imply $V^{(\hat{\delta})*}\rightarrow V^{*}$
as $\hat{\delta}\rightarrow\infty$. The detailed proof is provided
in Section \ref{subsec:Proof-of-Theorem-approximate}.
\end{IEEEproof}
\begin{algorithm}[th]
\caption{\label{algorithm_of_RVI_structure}RVI based on Threshold Sturcture}

\small

\begin{algorithmic}[1]

\STATE \textbf{Initialization}: $h^{(\hat{\delta})}_{-1}(\bm{s})=\infty,h^{(\hat{\delta})}_{0}(\bm{s})=0,\forall\bm{s}\in\mathcal{S^{\hat{\delta}}}$,
the number of iterations $n=0$, tolerance $\bar{\epsilon}>0$.

\WHILE{ $\max_{\bm{s}\in\mathcal{S^{\hat{\delta}}}}\{\left|h^{(\hat{\delta})}_{n}(\bm{s})-h^{(\hat{\delta})}_{n-1}(\bm{s})\right|\}>\bar{\epsilon}$
}

\STATE $\theta^{*}(\delta)=1$ for all $\bm{s}\in\mathcal{S^{\hat{\delta}}}$.

\FOR {$\bm{s}\in$ $\mathcal{S}^{(\hat{\delta})}$ }

\IF {$\theta\geq\theta^{*}(\delta)$ }

\STATE $u^{*}=1$.

\ELSE

\STATE $u^{*}=\operatorname*{arg\,min}_{u\in\{0,1\}}\{C(\bm{s},u)+\sum_{\bm{s}'\in\mathcal{S}^{\hat{\delta}}}P^{\hat{\delta}}_{\bm{s,s'}}(u)h^{(\hat{\delta})}_{n}(\bm{s}')\}.$

\IF {$u^{*}=1$ }

\STATE $\theta^{*}(\delta)=\theta$ .

\ENDIF

\ENDIF

\STATE $Q^{(\hat{\delta})}_{n+1}(\bm{s},u^{*})=C(\bm{s},u^{*})+\sum_{\bm{s}'\in\mathcal{S}^{\hat{\delta}}}P^{\hat{\delta}}_{\bm{s,s'}}(u^{*})h^{(\hat{\delta})}_{n}(\bm{s}').$

\STATE $h^{(\hat{\delta})}_{n+1}(\bm{s})=Q^{(\hat{\delta})}_{n+1}(\bm{s},u^{*})-\min_{u}Q^{(\hat{\delta})}_{n+1}(\bm{s}_{0},u)$.

\ENDFOR

\STATE $n=n+1$ .

\ENDWHILE

\end{algorithmic}
\end{algorithm}

Theorem \ref{thm:Approximat} justifies computing
the index on a finite-state approximation of the single-arm MDP. Based
on this approximation, Algorithm \ref{algorithm_of_RVI_structure}
exploits the threshold structure to implement RVI more efficiently
and to obtain the Q-function needed for index computation. For the
approximate MDP, we denote $h^{(\hat{\delta})}(\bm{s})$ as the relative
cost function of state $\bm{s}$. It satisfies the Bellman optimality
equation as follows:
\begin{align}
V^{*}+h^{(\hat{\delta})}(\bm{s}) & =\min_{u\in\{0,1\}}Q^{(\hat{\delta})}(\bm{s},u),\\
Q^{(\hat{\delta})}(\bm{s},u) & =C(\bm{s},u)+\sum_{\bm{s}'\in\mathcal{S}^{\hat{\delta}}}P^{\hat{\delta}}_{\bm{s,s'}}(u)h^{(\hat{\delta})}(\bm{s}'),
\end{align}
where $V^{*}$ is the optimal average cost which is a constant for
all states in $\mathcal{S^{\hat{\delta}}}$ and $P^{\hat{\delta}}_{\bm{s,s'}}(u)$
is the transition probability from state $\bm{s}$ to $\bm{s'}$ on
state space $\mathcal{S^{\hat{\delta}}}$ with action $u$. The RVI
algorithm begins with a random initialization of $h^{(\hat{\delta})}_{0}(\bm{s}),\forall\bm{s}\in\mathcal{S^{\hat{\delta}}}$,
and then updates the Q-function and relative cost function continuously
until convergence. Actually, in Theorem \ref{thm:threshold}, we find
that for fixed AoI, the optimal policy of MDP possesses a threshold
structure with respect to belief and we can reduce the computational
complexity in the process of policy improvement accordingly. Particularly,
if the optimal action for a certain state $\bm{s}=(\delta,\theta_{1})$
is updating, then it is still optimal to schedule for the states $\bm{s}=(\delta,\theta_{2})$
where $\theta_{2}\geq\theta_{1}$. Using such property of optimal
policy, we can determine the optimal action of some state immediately
(Lines 5-6 in Algorithm \ref{algorithm_of_RVI_structure}). Otherwise,
the optimal action only can be determined by comparison of $Q^{(\hat{\delta})}_{n}(\bm{s},u)$
(Line 8 in Algorithm \ref{algorithm_of_RVI_structure}). When $n\rightarrow\infty$,
$Q^{(\hat{\delta})}_{n+1}(\bm{s},u)$ will converge to $Q^{(\hat{\delta})}(\bm{s},u)$.

Let $W(\bm{s})$ denote the Whittle's index of state $\bm{s}$. According
to its definition, we can obtain $W(s)$ by finding the solution to
$Q^{(\hat{\delta})}(\bm{s},1)=Q^{(\hat{\delta})}(\bm{s},0)$. However,
it is challenging to solve this equation in closed-form. To address
this problem, we define a function $F(W,\bm{s})=Q^{(\hat{\delta})}_{W}(\bm{s},1)-Q^{(\hat{\delta})}_{W}(\bm{s},0)$
and obtain $W(s)$ by finding the zero point of the function. From
the process of proving Theorem \ref{thm:threshold}, we can find that
$F(W,\bm{s})$ is a non-decreasing function with respect to $W$.
Then, we first determine the search region $[W_{LB},W_{UB}]$, where
$W_{LB}$ and $W_{UB}$ are the lower bound and the upper bound of
the region, respectively. Additionally, the length $d$ of the region
is the distance between $W_{LB}$ and $W_{UB}$ (Lines 3-6 in Algorithm
\ref{alg:calculate-WI}). After determining the search region, we
use Ternary Search to find the zero point (Lines 7-14 in Algorithm
\ref{alg:calculate-WI}). In the process of calculating $F(W,\bm{s})$,
we can use Algorithm \ref{algorithm_of_RVI_structure} to obtain the
value of $Q^{(\hat{\delta})}_{W}(\bm{s},u)$.{} Repeating
this procedure for each state $\bm{s}\in\mathcal{S}^{(\hat{\delta})}$
gives the corresponding Whittle index $W(\bm{s})$ in the general
case, as summarized in Algorithm \ref{alg:calculate-WI}.

\begin{algorithm}[t]
\caption{\label{alg:calculate-WI}Computing the Whittle's index}

\begin{algorithmic}[1]

\STATE \textbf{Initialization}: Set $W_{LB}=0$ and $W_{UB}=W_{LB}+d$.

\FOR {$\bm{s}$ in $\mathcal{S}^{(\hat{\delta})}$ }

\WHILE{ $F(W_{UB},\bm{s})<0$ }

\STATE $W_{LB}=W_{UB}$.

\STATE $W_{UB}=W_{LB}+d$.

\ENDWHILE

\WHILE {$W_{UB}-W_{LB}>\xi$}

\STATE $W'=\frac{W_{UB}+W_{LB}}{2}$

\IF{$F(W',\bm{s})<0$}

\STATE $W_{LB}=W',W_{UB}=W_{UB}$ .

\ELSE

\STATE$W_{LB}=W_{LB}$, $W_{UB}=W'$.

\ENDIF

\ENDWHILE

\STATE $W(\bm{s})=\frac{W_{LB}+W_{UB}}{2}.$

\ENDFOR

\end{algorithmic}
\end{algorithm}

Now, we are ready to propose the Whittle's index-based scheduling
policy.{} With these state-wise indices, the Whittle's
index-based scheduling policy is implemented as follows. At the beginning
of each time slot, the destination computes the Whittle's index of
each IoT device based on the AoI of each device and the belief of
the channel state. Then, the destination schedules the $K$ IoT devices
with the largest Whittle's index.
\begin{figure}[!t]
\centering \includegraphics[width=0.99\columnwidth]{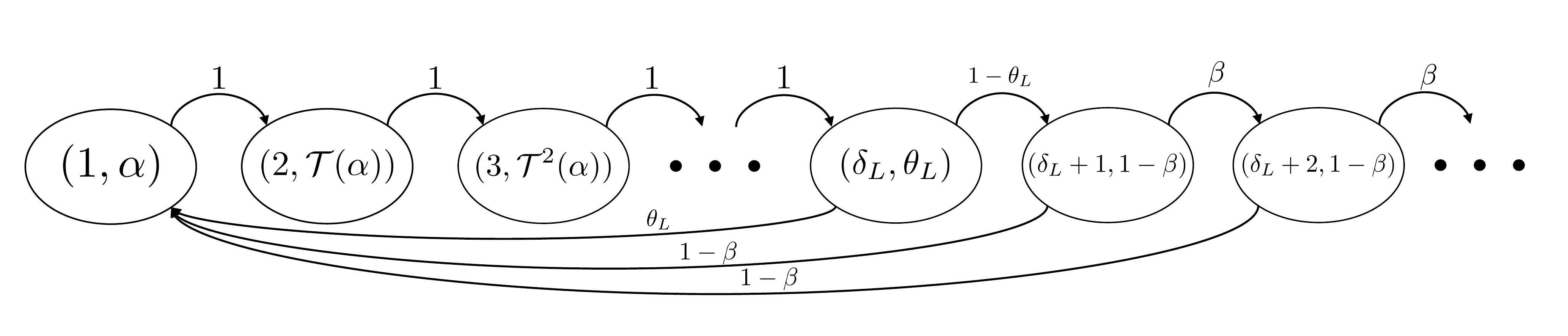}\caption{The state transition of the DTMC induced by policy $h$.}
\label{fig:The-state-transition-heuristic}
\end{figure}
\vspace{-0.5cm}

\subsection{Whittle-Like Index Policy}

In this subsection, we address the computational challenge associated
with calculating the Whittle's index based on RVI, a process that
becomes prohibitively complex as the state space grows. Our approach
is to derive a closed-form expression for a Whittle-like index, thereby
significantly reducing the computational overhead.

Following from Theorem \ref{thm:threshold}, we observe that the optimal
policy for the MDP exhibits a threshold structure with respect to
the belief state. Specifically, an arm is optimally scheduled when
its belief exceeds a certain threshold, which is dependent on both
the AoI and extra cost. In the decoupled problem, this threshold policy
defines the state $\bm{s}=(\delta_{L},\theta_{L})$ as the point at
which the arm is first activated. However, a challenge arises when
the channel remains in a BAD state after an arm is scheduled. In this
case, the system transitions to a new state $\bm{s}'=(\delta_{L}+1,1-\beta)$,
but the optimal action for this new state may not be immediately known
based on the threshold policy alone. To address this challenge, we
investigate a heuristic policy, denoted as $\pi^{\mathrm{H}}$,
which dictates that the arm continues to be scheduled from state $\bm{s}=(\delta_{L},\theta_{L})$
until a successful transmission occurs. This heuristic induces an
irreducible and aperiodic Discrete Time Markov Chain (DTMC), as illustrated
in Fig. \ref{fig:The-state-transition-heuristic}.

We denote $\varphi_{j}$ as the stationary distribution of the states
in this DTMC. The following lemma provides the closed-form expression
for this stationary distribution.
\begin{lem}
\label{lem:The-stationary-distribution}The stationary distribution
of states in Fig. \ref{fig:The-state-transition-heuristic} is given
by:
\begin{equation}
\varphi_{j}(\delta_{L},\theta_{L})=\begin{cases}
\frac{1-\beta}{(1-\beta)\delta_{L}+1-\theta_{L}}, & 1\leq j\leq\delta_{L},\\
\frac{(1-\beta)(1-\theta_{L})\beta^{j-\delta_{L}-1}}{(1-\beta)\delta_{L}+1-\theta_{L}}, & j\geq\delta_{L}+1.
\end{cases}\label{eq:stationary-distribution-heuristic}
\end{equation}
\end{lem}
\begin{IEEEproof}
For the DTMC in Fig. \ref{fig:The-state-transition-heuristic},
the balance equations give $\varphi_{j}=\varphi_{j-1}$ for $1<j\leq\delta_{L}$
and
\begin{equation}
\varphi_{j}=\varphi_{\delta_{L}}(1-\theta_{L})\beta^{j-\delta_{L}-1},\quad j\geq\delta_{L}+1.
\end{equation}
Using $\sum^{\infty}_{j=1}\varphi_{j}=1$, we obtain $\varphi_{\delta_{L}}=(1-\beta)/((1-\beta)\delta_{L}+1-\theta_{L})$.
Substituting this value into the two balance relations yields the
stated stationary distribution.
\end{IEEEproof}
Leveraging the stationary distribution from Lemma \ref{lem:The-stationary-distribution},
the average cost $\bar{C}_{\mathrm{H}}(\delta_{L},\theta_{L},W)$
of the heuristic policy $\pi^{\mathrm{H}}$ is calculated
as follows:
\begin{align}
 & \bar{C}_{\mathrm{H}}(\delta_{L},\theta_{L},W)=\sum^{\infty}_{j=1}j\varphi_{j}+W\sum^{\infty}_{j=\delta_{L}}\varphi_{j}\nonumber \\
 & =\varphi_{\delta_{L}}[\frac{\delta_{L}(\delta_{L}+1)}{2}+\frac{(1-\theta_{L})\beta}{(1-\beta)^{2}}\nonumber \\
 & +\frac{(1-\theta_{L})(\delta_{L}+1)}{1-\beta}+\frac{2-\theta_{L}-\beta}{1-\beta}W].\label{eq:average-cost-heuristic}
\end{align}
For clarity, we define $F_{1}(\delta_{L},\theta_{L})=\frac{1-\beta}{(1-\beta)\delta_{L}+1-\theta_{L}}$
$\left[\frac{\delta_{L}(\delta_{L}+1)}{2}+\frac{(1-\theta_{L})\beta}{(1-\beta)^{2}}+\frac{(1-\theta_{L})(\delta_{L}+1)}{1-\beta}\right]$
and $F_{2}(\delta_{L},\theta_{L})=\frac{2-\theta_{L}-\beta}{(1-\beta)\delta_{L}+1-\theta_{L}}$.
Then, the average cost can be expressed as:
\begin{equation}
\bar{C}_{\mathrm{H}}(\delta_{L},\theta_{L},W)=F_{1}(\delta_{L},\theta_{L})+F_{2}(\delta_{L},\theta_{L})W.\label{eq:average-cost-heuristic-simplify}
\end{equation}
For state $\bm{s}=(\delta_{L},\theta_{L})$ where the optimal action
is to be active, if an arm with such a state is not scheduled, it
will update to $\bm{s}'=(\delta_{L}+1,\mathcal{T}(\theta_{L})).$
Similar to the methodology in Theorem \ref{thm9:special_case}, we
construct the equation $\bar{C}_{\mathrm{H}}(\delta_{L},\theta_{L},W)=\bar{C}_{\mathrm{H}}(\delta_{L}+1,\mathcal{T}(\theta_{L}),W)$.
Solving this equation allows us to obtain the Whittle-like index,
which is given in the following theorem.
\begin{thm}
For any state $\bm{s}=(\delta,\theta)$, the corresponding Whittle-like
index can be calculated as follows:
\begin{equation}
W_{L}(\delta,\theta)=\frac{F_{1}(\delta+1,\mathcal{T}(\theta))-F_{1}(\delta,\theta)}{F_{2}(\delta,\theta)-F_{2}(\delta+1,\mathcal{T}(\theta))}.\label{eq:whittle_like_index}
\end{equation}
\end{thm}
\begin{IEEEproof}
The Whittle-Like index is obtained by solving the equation $\bar{C}_{\mathrm{H}}(\delta,\theta,W)=\bar{C}_{\mathrm{H}}(\delta+1,\mathcal{T}(\theta),W)$
for $W$. By substituting \eqref{eq:average-cost-heuristic}, we can
derive the closed-form expression for $W_{L}(\delta,\theta)$.
\end{IEEEproof}

This closed-form expression of the Whittle-like index enables the
implementation of a computationally efficient scheduling policy. Similar
to the classical Whittle's index policy, our approach is to compute
the Whittle-like index for each arm and then schedule $K$ IoT devices
with the largest Whittle-like indices. Since the Whittle-like index
can be directly obtained from the derived expression, this method
significantly reduces the computational complexity compared to the
Whittle's index policy.

\textbf{Remark:}{} The Whittle's
index policy and the Whittle-like index policy provide different implementation
tradeoffs. The Whittle's index policy is preferable when the system
size is moderate and the channel transition probabilities are stable,
so that the state-wise indices can be computed offline or updated
infrequently. In contrast, the Whittle-like index policy is more suitable
for large-scale, resource-constrained, or real-time IoT systems because
its closed-form expression can be evaluated directly from the current
AoI and belief state.

To quantify this tradeoff, let $N=|\mathcal{S}^{\hat{\delta}}|$
denote the number of states in the truncated single-arm MDP. Since
$|\mathcal{B}_{\delta}|=O(\delta)$, we have $N=O(\hat{\delta}^{2})$.
For the Whittle's index policy, the offline phase computes the state-wise
indices by Algorithms \ref{algorithm_of_RVI_structure} and \ref{alg:calculate-WI}.
Since each state-action pair has only a constant number of possible
next states in the truncated single-arm MDP, one RVI iteration has
time complexity $O(N)$, and the overall index computation has time
complexity $O(I_{\mathrm{S}}I_{\mathrm{RVI}}N^{2})$, where $I_{\mathrm{RVI}}$
and $I_{\mathrm{S}}$ denote the numbers of RVI and one-dimensional
search iterations, respectively. The offline storage for the value
functions and index tables is $O(MN)$ for $M$ device-specific channel
models. If multiple devices share the same channel parameters, the
corresponding index table can be reused, reducing the offline storage.
In the online phase, the scheduler maintains the current AoI and belief
state of each device, which requires $O(M)$ memory, and selects the
$K$ largest indices among $M$ devices, which can be implemented
in $O(M\log K)$ time. Thus, the total memory footprint of the Whittle's
index policy is $O(MN+M)=O(MN)$ in the device-specific case, or $O(N+M)$
when a shared index table can be reused.

In contrast, the Whittle-like index policy does not
have an offline index-computation phase, because its index is evaluated
from a closed-form expression. In the online phase, it computes one
index for each device and selects the $K$ largest indices, leading
to $O(M\log K)$ time complexity. Its memory footprint is $O(M)$,
since it only maintains the current AoI and belief states. Therefore,
the Whittle-like index policy is generally more suitable for large-scale
and real-time IoT implementations, while the Whittle's index policy
remains useful when offline computation and index storage are affordable.

\begin{algorithm}[tb]
\caption{\label{alg:Finding-the-extra}Finding an approximation to the extra
cost $W^{*}$}

\begin{algorithmic}[1]

\STATE \textbf{Initialization}: Set $W_{LB}=0$ and $W_{UB}=W_{LB}+d$.

\WHILE{ $J(W_{LB})\leq J(W_{UB})$ }

\STATE $W_{LB}=W_{UB}$.

\STATE $W_{UB}=W_{LB}+d$.

\ENDWHILE

\STATE $W_{LB}=W_{LB}-d,W_{UB}=W_{UB}$.

\WHILE {$W_{UB}-W_{LB}>\xi$}

\STATE $m_{1}=W_{LB}+\frac{W_{UB}-W_{LB}}{3}$, $m_{2}=W_{UB}-\frac{W_{UB}-W_{LB}}{3}$.

\IF{$J(m_{1})>J(m_{2})$}

\STATE $W_{LB}=W_{LB}$, $W_{UB}=m_{2}$.

\ELSE

\STATE$W_{LB}=m_{1}$, $W_{UB}=W_{UB}$.

\ENDIF

\ENDWHILE

\STATE $W^{\dagger}=\frac{W_{LB}+W_{UB}}{2}.$

\end{algorithmic}
\end{algorithm}

\subsection{Optimal Policy for The Relaxed Problem}

In this subsection, we study the optimal scheduling policy for the
relaxed partially observable RMAB problem, where the strict constraint
on the number of scheduled devices in each time slot is relaxed as
follows:
\begin{equation}
\mathbb{E}\left[\lim_{T\rightarrow\infty}\frac{1}{T}\sum^{M}_{i=1}\sum^{T}_{t=1}u_{i}(t)\right]=K.\label{eq:relax-cons}
\end{equation}
It is obvious that the performance of the optimal scheduling policy
for the above relaxed partially observable RMAB problem provides a
lower bound for the Whittle\textquoteright s index policy under the
strict constraint.{} This lower bound will be used
to assess how close the proposed index-based policies are to the relaxed
optimum.

We let $\bar{V}^{*}$ denote the minimum expected average cost that
can be obtained under the relaxed constraint given in (\ref{eq:relax-cons}).
According to the Lagrangian multiplier theorem, we have \cite{whittleRestlessBanditsActivity1988}
\begin{equation}
\bar{V}^{*}=\sup_{W}\left\{ \sum^{M}_{i=1}V^{*}_{i}(W)-WK\right\} ,\label{eq23:optimal_W}
\end{equation}
where $V^{*}_{i}(W)$ is the optimal average cost of the sub-problem
with extra cost $W$ for device $i$.\footnote{The value of $V^{*}_{i}(W)$ can be obtained via the RVI algorithm.}
We denote $W^{*}$ as the extra cost that achieves the supremum in
\eqref{eq23:optimal_W}. The optimal policy for the relaxed partially
observable RMBP problem with constraint (\ref{eq:relax-cons}) is
to schedule, in each time slot, the devices whose Whittle's index
is greater than $W^{*}$ \cite{5605371}. In general, it is challenging
to derive $W^{*}$ in closed-form. However, the properties of $V^{*}_{i}(W)$
allow us to design an efficient algorithm for computing $W^{*}$.
Similar to \cite{5605371} and \cite{whittleRestlessBanditsActivity1988},
$V^{*}_{i}(W)$ can be proved to be concave increasing in $W$(A
proof of this property is provided in Section \ref{subsec:Proof-for-ViW}).
Hence, $\sum^{M}_{i=1}V^{*}_{i}(W)-WK$ is a concave and unimodal
function. We let $J(W)=\sum^{M}_{i=1}V^{*}_{i}(W)-WK$ and propose
an algorithm with two steps to find $W^{*}$, as depicted in Algorithm
\ref{alg:Finding-the-extra}.{} Thus, Algorithm \ref{alg:Finding-the-extra}
can be viewed as a one-dimensional search for the Lagrange multiplier
that yields the relaxed-problem lower bound. The first step (lines
2\textasciitilde 6 in Algorithm \ref{alg:Finding-the-extra}) is
to determine the search region $[W_{LB},W_{UB}]$, where $W_{LB}$
and $W_{UB}$ are the lower bound and the upper bound of the region,
respectively. Both $W_{LB}$ and $W_{UB}$ are increased until $J(W_{LB})>J(W_{UB})$
and then we determine the search region such that it includes $W^{*}$.
The length of the search region, denoted by $d$, is fixed in this
step. The second step (lines 7\textasciitilde 14 in Algorithm \ref{alg:Finding-the-extra})
is to find $W^{*}$ by using Ternary Search in the region of $[W_{LB},W_{UB}]$.
We continually narrow the search region including $W^{*}$. By letting
the length of the search region small enough, i.e., $W_{UB}-W_{LB}\leq\xi,\forall\xi>0$,
we can find a $W^{\dagger}$ that is arbitrarily close to $W^{*}$
and hence use $W^{\dagger}$ to approximate $W^{*}$.

\subsection{Myopic Policy}

We propose a low-complexity myopic policy that only optimizes the
expected immediate reward and ignores the impact of the current action
on the future reward.{} Since it does not account for
the future evolution of the belief and AoI states, this policy provides
a simple baseline for evaluating the gain of the proposed index-based
policies. In particular, we define the expected post-action age $\tilde{\delta}_{i}(t)$
of device $i$ as the expected age after action $u_{i}(t)=1$ taken
in slot $t$. Given the current state $\bm{s}_{i}(t)=(\delta_{i}(t),\theta_{i}(t))$,
the expected post-action age is given by $\tilde{\delta}_{i}(t)=\theta_{i}(t)+(\delta_{i}(t)+1)(1-\theta_{i}(t))=\delta_{i}(t)+1-\delta_{i}(t)\theta_{i}(t)$.
Accordingly, we define the expected age difference of device $i$
in slot $t$ as $\delta_{i}(t)-\tilde{\delta}_{i}(t)=\delta_{i}(t)\theta_{i}(t)-1$.
Then, in each time slot, the destination calculates the expected age
difference for each device and schedules $K$ devices with the largest
expected age difference.

\section{Technical Proofs of Main Results\label{sec:technical-proofs}}

This section collects the detailed proofs of the main theoretical
results. 

\subsection{Proof of Lemma \ref{lem:lemma2:concavity} \label{subsec:Proof-of-Lemma2}}

Similar to the proof of Lemma \ref{lem:Lemma1:discount_monitonicity},
we use mathematical induction to prove Lemma \ref{lem:lemma2:concavity}.
We first show the monotonicity of the value function in the belief.
Let $\bm{s}_{1}=(\delta,\theta_{1}),$ $\bm{s}_{2}=(\delta,\theta_{2})\in\mathcal{S},$
where $\theta_{1}\leq\theta_{2}$. The statement holds at $n=0$ under
the zero initialization. Suppose that $V_{\mu,n}(\bm{s}_{1})\geq V_{\mu,n}(\bm{s}_{2})$
holds at iteration $n$, and verify whether it still holds at iteration
$n+1$. When $u=0$, we have
\begin{align}
Q_{\mu,n+1}(\bm{s}_{1},0) & =\delta+\mu V_{\mu,n}(\delta+1,\mathcal{T}(\theta_{1}))\nonumber \\
 & \geq\delta+\mu V_{\mu,n}(\delta+1,\mathcal{T}(\theta_{2}))=Q_{\mu,n+1}(\bm{s}_{2},0).
\end{align}
When $u=1$, the function $Q_{\mu,n+1}(\delta,\theta,1)$ is affine
in $\theta$ with slope $\mu[V_{\mu,n}(1,\alpha)-V_{\mu,n}(\delta+1,1-\beta)]$.
Since Lemma \ref{lem:Lemma1:discount_monitonicity} and the induction
hypothesis imply $V_{\mu,n}(1,\alpha)\leq V_{\mu,n}(\delta+1,\alpha)\leq V_{\mu,n}(\delta+1,1-\beta)$,
this slope is non-positive. Hence,
\begin{align}
Q_{\mu,n+1}(\bm{s}_{1},1)\geq Q_{\mu,n+1}(\bm{s}_{2},1).
\end{align}
Therefore, we have $V_{\mu,n+1}(\bm{s}_{1})\geq V_{\mu,n+1}(\bm{s}_{2})$.
When $n\rightarrow\infty,$ we conclude that the value function $V_{\mu}(\bm{s})$
is non-increasing with respect to $\theta$.

Next, we show the concavity of the value function in $\theta$. To
prove such property, it is sufficient to prove that for all $n\geq0,$
\begin{equation}
V_{\mu,n}(\delta,a\theta_{1}+(1-a)\theta_{2})\geq aV_{\mu,n}(\delta,\theta_{1})+(1-a)V_{\mu,n}(\delta,\theta_{2}),\label{eq:concave_equation}
\end{equation}
where $a\in[0,1]$. This inequality holds at $n=0$ under the zero
initialization. Suppose that \eqref{eq:concave_equation} holds at
iteration $n$, and consider iteration $n+1$ by examining the Q-functions.
For convenience, we let $\bm{s}_{1}=(\delta,\theta_{1})$, $\bm{s}_{2}=(\delta,\theta_{2})$,
and $\bm{s}_{3}=(\delta,a\theta_{1}+(1-a)\theta_{2})$. For $u=0$,
the linearity of $\mathcal{T}(\theta)$ and the induction hypothesis
give
\begin{align}
Q_{\mu,n+1}(\bm{s}_{3},0)= & \delta+\mu V_{\mu,n}(\delta+1,\mathcal{T}(a\theta_{1}+(1-a)\theta_{2}))\nonumber \\
= & \delta+\mu V_{\mu,n}(\delta+1,a\mathcal{T}(\theta_{1})+(1-a)\mathcal{T}(\theta_{2}))\nonumber \\
\geq & \delta+\mu aV_{\mu,n}(\delta+1,\mathcal{T}(\theta_{1}))\nonumber \\
 & +\mu(1-a)V_{\mu,n}(\delta+1,\mathcal{T}(\theta_{2}))\nonumber \\
= & aQ_{\mu,n+1}(\bm{s}_{1},0)+(1-a)Q_{\mu,n+1}(\bm{s}_{2},0)\nonumber \\
\geq & aV_{\mu,n+1}(\bm{s}_{1})+(1-a)V_{\mu,n+1}(\bm{s}_{2}).
\end{align}
For $u=1$, $Q_{\mu,n+1}(\delta,\theta,1)$ is affine in $\theta$.
Hence,
\begin{align}
Q_{\mu,n+1}(\bm{s}_{3},1)= & aQ_{\mu,n+1}(\bm{s}_{1},1)+(1-a)Q_{\mu,n+1}(\bm{s}_{2},1)\nonumber \\
\geq & aV_{\mu,n+1}(\bm{s}_{1})+(1-a)V_{\mu,n+1}(\bm{s}_{2}).
\end{align}
According to (\ref{eq:Bellman in VIA}), we can derive that $V_{\mu,n+1}(\bm{s}_{3})\geq aV_{\mu,n+1}(\bm{s}_{1})+(1-a)V_{\mu,n+1}(\bm{s}_{2}).$
When $n\rightarrow\infty,$ we conclude that $V_{\mu}(\bm{s}_{3})\geq aV_{\mu}(\bm{s}_{1})+(1-a)V_{\mu}(\bm{s}_{2}).$

\subsection{Proof of Theorem \ref{thm:existence_optimal}\label{subsec:Proof-of-Theorem1}}

\begin{figure}[t]
\centering

\includegraphics[width=0.48\textwidth]{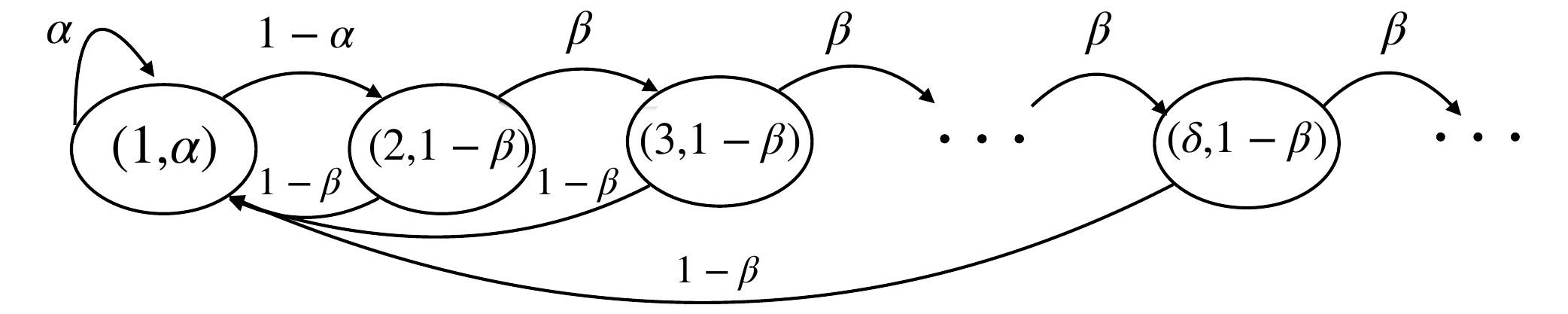}

\caption{The state transition of the DTMC induced by policy $f$.\label{fig:DTMC}}
\end{figure}
Let $h_{\mu}(\bm{s})=V_{\mu}(\bm{s})-V_{\mu}(\bm{s}_{0})$ be the
relative cost function, where $\bm{s}_{0}$ is a reference state.
According to \cite{sennottAverageCostOptimal1989}, Theorem \ref{thm:existence_optimal}
can be proved by verifying the following two conditions.
\begin{enumerate}
\item There exists a deterministic stationary policy of the average cost
MDP such that it induces an irreducible and aperiodic DTMC and its
average cost is finite.
\item There exists a non-negative $L$ such that $h_{\mu}(\bm{s})\geq-L$
for all $\bm{s}$ and $\mu$.
\end{enumerate}
We first verify condition 1. We consider a deterministic stationary
policy $f$ which schedules the device to update in each time slot.
The induced DTMC by policy $f$ is shown in Fig \ref{fig:DTMC}. It
is obvious that the induced DTMC is irreducible and aperiodic. We
let $\varphi_{\delta}$ be the stationary distribution of state $\bm{s}=(\delta,\theta)$,
which is given by
\begin{align}
\varphi_{\delta} & =\varphi_{1}(1-\alpha)\beta^{\delta-2},\text{ for all }\delta=2,3,\ldots,
\end{align}
where $\varphi_{1}=\frac{1-\beta}{2-\alpha-\beta}.$ Hence, the average
cost of the MDP under policy $f$ can be calculated as:
\begin{align}
\bar{C} & =\sum^{\infty}_{\delta=1}\varphi_{\delta}(\delta+W)=W+\frac{(2-\beta)(\beta^{2}+1-\alpha)}{(1-\beta)(2-\alpha-\beta)}<\infty.
\end{align}

We next verify condition 2. We set the reference state as $\bm{s}_{0}=(1,\alpha)$.
 By Lemma \ref{lem:lemma2:concavity}, $V_{\mu}(\delta,\theta)$
is non-increasing in $\theta$ for any given $\delta$, and by Lemma
\ref{lem:Lemma1:discount_monitonicity}, $V_{\mu}(\delta,\theta)$
is non-decreasing in $\delta$ for any given $\theta$. Therefore,
for any state $\bm{s}=(\delta,\theta)$, we have $V_{\mu}(\delta,\theta)\ge V_{\mu}(1,\theta)\ge V_{\mu}(1,\alpha)=V_{\mu}(\bm{s}_{0})$.
Letting $L=0$ gives $h_{\mu}(\bm{s})=V_{\mu}(\bm{s})-V_{\mu}(\bm{s}_{0})\geq0$.
Hence, condition 2 also holds, and the existence of an average cost
optimal stationary deterministic policy follows.

\subsection{Proof of Theorem \ref{thm:threshold}\label{subsec:Proof-of-Theorem2}}

We first establish the corresponding structural
property for the discounted-cost problem.
\begin{lem}
\label{lem:threshold}For the positively correlated Markov channel,
the following properties hold \cite{yaoAgeOptimalLowPowerStatus2020}.

(a) For any $\delta$, $\theta$, and $0\leq k\leq\delta-1$,
\begin{align}
(1-\theta)W+\theta V_{\mu}(\delta,\mathcal{T}^{k}(\alpha))+(1-\theta)V_{\mu}(\delta,\mathcal{T}^{k}(1-\beta))\nonumber \\
\geq V_{\mu}(\delta,\mathcal{T}^{k+1}(\theta)).
\end{align}

(b) The discounted cost optimal policy corresponding to $V_{\mu}(\bm{s})$
is of threshold-type in $\theta$. That is, for any $\delta$ and
$\theta_{1}\leq\theta_{2}$, if it is optimal to take action $u=1$
at $(\delta,\theta_{1})$, then it is also optimal to take action
$u=1$ at $(\delta,\theta_{2})$.
\end{lem}
\begin{IEEEproof}
The lemma follows by value-iteration induction. With $V_{\mu,0}(\bm{s})=0$,
both properties hold at $n=0$. Suppose that they hold up to iteration
$n$. To prove the threshold property at $n+1$, compare the two state-action
value functions $Q_{\mu,n+1}(\delta,\theta,0)$ and $Q_{\mu,n+1}(\delta,\theta,1)$
over the extended belief interval $[0,1]$. By Lemma \ref{lem:lemma2:concavity},
the former is concave in $\theta$, while the latter is linear in
$\theta$. At the two boundary points, we have
\begin{align}
Q_{\mu,n+1}(\delta,0,0) & =\delta+\mu V_{\mu,n}(\delta+1,1-\beta),\label{eq:Q0}\\
Q_{\mu,n+1}(\delta,0,1) & =\delta+W+\mu V_{\mu,n}(\delta+1,1-\beta),\\
Q_{\mu,n+1}(\delta,1,0) & =\delta+\mu V_{\mu,n}(\delta+1,\alpha),\\
Q_{\mu,n+1}(\delta,1,1) & =\delta+W+\mu V_{\mu,n}(1,\alpha).\label{eq:Q1}
\end{align}
Since $Q_{\mu,n+1}(\delta,0,0)\leq Q_{\mu,n+1}(\delta,0,1)$ for $W\geq0$,
the boundary comparison, together with property (a) at iteration $n$,
characterizes the action preference by a belief threshold. If $0\leq W<\mu V_{\mu,n}(\delta+1,\alpha)-\mu V_{\mu,n}(1,\alpha)$,
the two Q-functions intersect at a belief threshold, above which transmission
is optimal. Otherwise, property (a) at iteration $n$ implies $Q_{\mu,n+1}(\delta,\theta,0)\leq Q_{\mu,n+1}(\delta,\theta,1)$
for all $\theta\in[0,1]$, so idling is optimal for all beliefs. Thus,
property (b) holds at $n+1$.

Given property (b) at iteration $n+1$, property (a) at $n+1$ follows
by considering the optimal actions at $(\delta,\mathcal{T}^{k}(\alpha))$
and $(\delta,\mathcal{T}^{k}(1-\beta))$. The action pair $(0,1)$
is excluded by the threshold structure, and the remaining cases follow
from Lemmas \ref{lem:Lemma1:discount_monitonicity} and \ref{lem:lemma2:concavity}
together with the induction hypothesis. Letting $n\rightarrow\infty$
yields both properties in Lemma \ref{lem:threshold}.
\end{IEEEproof}
We now prove the threshold structure of the average cost optimal policy
based on Lemma \ref{lem:threshold}. Let $\{\mu_{n}\}^{\infty}_{n=1}$
be a sequence of discount factors converging to $1$ and $\pi^{*}_{\mu_{n}}$
be the discounted cost optimal policy associated with $\mu_{n}$.
According to \cite{sennottAverageCostOptimal1989}, since the average
cost optimal policy exists by Theorem \ref{thm:existence_optimal},
there exists a subsequence $\{\mu'_{n}\}^{\infty}_{n=1}$ such that
an average cost optimal policy is a limit point of $\pi^{*}_{\mu'_{n}}$.
Because each discounted-cost optimal policy in this subsequence has
a threshold structure in $\theta$, the limiting average cost optimal
policy also has a threshold structure.

\subsection{Proof of Theorem \ref{thm:indexability}\label{subsec:Proof-of-Theorem8}}

We prove the indexability by showing that the passive set for any
given $\delta$ increases monotonically from the empty set to the
entire space of $\theta$, i.e., $B_{\delta}$, with the increasing
of the extra cost W. As in the proof of Lemma \ref{lem:threshold},
we extend the space of $\theta$ to $[0,1]$ for ease of exposition.
Based on the threshold structure  given in Theorem \ref{thm:threshold},
the indexability can be proved by showing the monotonicity of the
threshold. In other words, if the threshold is monotonically increasing
with the cost $W$ for $0\leq W<\mu V_{\mu,n}(\delta+1,\alpha)-\mu V_{\mu,n}(1,\alpha)$,
then the sub-problem is indexable.

To verify this, we focus on the discounted cost case. According to
Eqs. (\ref{eq:Q0})-(\ref{eq:Q1}) in the proof of Lemma \ref{lem:threshold},
it is easy to see that the optimal policy is to always transmit if
$W<0$ since $Q_{\mu}(\delta,0,0)>Q_{\mu}(\delta,0,1)$ and $Q_{\mu}(\delta,1,0)>Q_{\mu}(\delta,1,1)$.
Hence, the passive set is empty in this case. When $0\leq W<\mu V_{\mu}(\delta+1,\alpha)-\mu V_{\mu}(1,\alpha)$,
there is an intersection between two state-action value functions,
which is in fact the threshold. Suppose that $W_{1},W_{2}\in[0,\mu V_{\mu}(\delta+1,\alpha)-\mu V_{\mu}(1,\alpha))$
and $W_{1}\leq W_{2}$. The threshold associated with $W_{1}$ is
denoted by $\theta^{*}_{W_{1}}$, and we have $Q(\delta,\theta^{*}_{W_{1}},0)=Q(\delta,\theta^{*}_{W_{1}},1)$,
i.e.,
\begin{align}
 & \mu V_{\mu}(\delta+1,\mathcal{T}(\theta^{*}_{W_{1}}))\nonumber \\
= & W_{1}+\mu\theta^{*}_{W_{1}}V_{\mu}(1,\alpha)+\mu(1-\theta^{*}_{W_{1}})V_{\mu}(\delta+1,1-\beta).
\end{align}
When the extra cost increases to $W_{2}$, the optimal action at state
$(\delta,\theta^{*}_{W_{1}})$ is to idle since 
\begin{align}
 & Q(\delta,\theta^{*}_{W_{1}},0)\nonumber \\
= & \delta+\mu V(\delta+1,\mathcal{T}(\theta^{*}_{W_{1}}))\nonumber \\
< & \delta+W_{2}+\mu\theta^{*}_{W_{1}}V_{\mu}(1,\alpha)+\mu(1-\theta^{*}_{W_{1}})V_{\mu}(\delta+1,1-\beta)\nonumber \\
= & Q(\delta,\theta^{*}_{W_{1}},1).
\end{align}
Then, we have $\theta^{*}_{W_{1}}\leq\theta^{*}_{W_{2}}$. Finally,
when $W\geq\mu V_{\mu,n}(\delta+1,\alpha)-\mu V_{\mu,n}(1,\alpha)$,
the optimal action is to never transmit according to Lemma \ref{lem:threshold},
and the passive set is equal to the entire space of $\theta$. 

\subsection{Proof of Theorem \ref{thm9:special_case}\label{subsec:Proof-of-Theorem9}}

It is challenging to derive the Whittle's index in closed-form according
to its definition. However, we can dispense with this difficulty by
deriving the average cost of the threshold policy. In particular,
when $\alpha=1-\beta$, the state only consists of the AoI. As proved
in Lemma 1 of \cite{9238787}, the average-cost optimal policy is
of threshold-type in AoI. Then, we can use a Discrete Time Markov
Chain (DTMC) in Fig. \ref{fig:The-DTMC-of-thresholdpolicy} to model
the MDP constructed by the threshold policy with the threshold of
$n$.
\begin{figure}[t]
\centering\includegraphics[width=0.49\textwidth]{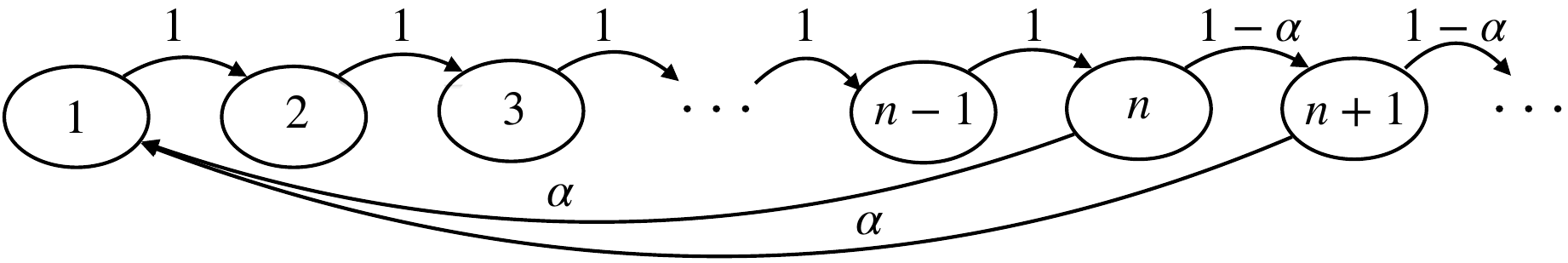}

\caption{The states transition under a threshold policy with the threshold
$n$.\label{fig:The-DTMC-of-thresholdpolicy}}
\end{figure}

We let $\eta_{j}(n)$ be the steady state probability of state $j$
under the threshold of $n$. The balance equations of the DTMC are
\begin{equation}
\eta_{j}(n)=\eta_{j-1}(n),\quad2\leq j\leq n\label{eq:Stationary_distribution_relationship1}
\end{equation}
and 
\begin{equation}
\eta_{j}(n)=(1-\alpha)\eta_{j-1}(n),\quad j\geq n+1.\label{eq:Stationary_distribution_relationship2}
\end{equation}
Since $\sum^{\infty}_{j=1}\eta_{j}(n)=1$, the steady state probability
of all states can be solved as
\begin{equation}
\eta_{j}(n)=\begin{cases}
\frac{\alpha}{n\alpha+1-\alpha}, & 1\leq j\leq n,\\
\frac{\alpha(1-\alpha)^{j-n}}{n\alpha+1-\alpha}, & j\geq n+1.
\end{cases}
\end{equation}
Then, the average cost of the MDP under the threshold policy can be
calculated as
\begin{align}
\bar{C}(n,W)= & \sum^{\infty}_{j=1}j\eta_{j}(n)+\sum^{\infty}_{j=n}W\eta_{j}(n)\label{eq:equation-of-average-cost}\\
= & \frac{n^{2}\alpha^{2}-n\alpha^{2}+2n\alpha-2\alpha+2}{2\alpha(n\alpha+1-\alpha)}+\frac{W}{\alpha n+1-\alpha}.\nonumber 
\end{align}
Then, the Whittle's index $W_{\delta}$ for state $\delta$ is the
intersection points of $\bar{C}(\delta,W)$ and $\bar{C}(\delta+1,W)$.
Particularly, by solving $\bar{C}(\delta,W)=\bar{C}(\delta+1,W)$,
we have 
\begin{equation}
W_{\delta}=\delta+\frac{\delta(\delta-1)}{2}\alpha.
\end{equation}

\subsection{Proof of Theorem \ref{thm:Approximat}\label{subsec:Proof-of-Theorem-approximate}}

For the approximate MDP with bound $\hat{\delta}$, let $V^{\hat{\delta}}_{\mu}(\bm{s})$
denote the minimum expected total $\mu$-discounted cost and $h^{\hat{\delta}}_{\mu}(\bm{s})=V^{\hat{\delta}}_{\mu}(\bm{s})-V^{\hat{\delta}}_{\mu}(\bm{s_{0}})$
denote the relative cost function. According to \cite{Sennott_on_computing_average_cost},
it suffices to verify: 1) There exists a nonnegative $L$ and a nonnegative
finite function $F(\cdot)$ on $\mathcal{S}$ such that $-L\leq h^{\hat{\delta}}_{\mu}(\bm{s})\leq F(\bm{s})$
for all $\bm{s}\in\mathcal{S}^{\hat{\delta}}$, where $\hat{\delta}=M+1,M+2,...$
and $\mu\in(0,1)$; and 2) $\limsup_{\hat{\delta}\rightarrow\infty}V^{(\hat{\delta})*}\leq V^{*}$.

To verify condition 1, we consider a policy $g$ that schedules the
user with equal probability in each time slot. We let $\bar{C}^{\hat{\delta}}_{\bm{s,s_{0}}}(g)$
and $\bar{C}_{\bm{s,s_{0}}}(g)$ denote the expected cost of the first
passage from state $\bm{s}$ to $\bm{s}_{0}$ by applying policy $g$
to the approximate MDP and the original MDP, respectively. Similar
to Section \ref{subsec:Proof-of-Theorem1}, we set $\bm{s_{0}}=(1,\alpha)$
and have $L=0$ and $\bar{C}_{\bm{s,s_{0}}}(g)<\infty$. Moreover,
since $h^{\hat{\delta}}_{\mu}(\bm{s})$ is the relative expected cost
under an optimal policy, we have $h^{\hat{\delta}}_{\mu}(\bm{s})\leq\bar{C}^{\hat{\delta}}_{\bm{s,s_{0}}}(g)$.
Then, we can choose $\bar{C}_{\bm{s,s_{0}}}(g)$ as $F(\bm{s})$ if
$\bar{C}^{\hat{\delta}}_{\bm{s,s_{0}}}(g)\leq\bar{C}_{\bm{s,s_{0}}}(g)$.
To show this, we first present the transition probability $P^{\hat{\delta}}_{\bm{s,s'}}(u)$
from state $\bm{s}$ to $\bm{s}'$ on state space $\mathcal{S^{\hat{\delta}}}$
with action $u$ as follows:
\begin{equation}
P^{\hat{\delta}}_{\bm{s,s'}}(u)=P_{\bm{s,s'}}(u)+\sum_{\bm{r}\in\mathcal{S\setminus S}^{\hat{\delta}}}P_{\bm{s,r}}(u)\mathbf{1}(v(\bm{r})=\bm{s'}),
\end{equation}
where  $P_{\bm{s,s'}}(u)$ and $P_{\bm{s,r}}(u)$ are the transition
probabilities of the original MDP, $\mathbf{1}(\cdot)$ is the indicator
function, and $v(\bm{r})$ is an association function that associates
a state $\bm{r}\in\mathcal{S}\setminus\mathcal{S}^{\hat{\delta}}$
to a state $\bm{s'}\in\mathcal{S}^{\hat{\delta}}$. Specifically,
$v((\delta,\theta))=(\min\{\delta,\hat{\delta}\},\theta+(\mathcal{T}^{\widehat{\delta}-1}(\alpha)-\theta){\bf 1}(\mathcal{T}^{\hat{\delta}-2}(1-\beta)<\theta<\mathcal{T}^{\hat{\delta}-1}(\alpha))$.
Moreover, we can verify that $\bar{C}_{(\delta,\theta)\bm{,s_{0}}}(g)$
is non-decreasing in $\delta$ with given $\theta$ and non-increasing
in $\theta$ with given $\delta$ in a similar way to the proof of
Lemmas \ref{lem:Lemma1:discount_monitonicity} and \ref{lem:lemma2:concavity}.
Then, we obtain that
\begin{align}
 & \sum_{\bm{s'}\in\mathcal{S}^{\hat{\delta}}}P^{\hat{\delta}}_{\bm{s,s'}}(u)\bar{C}_{\bm{s',s_{0}}}(g)\nonumber \\
= & \sum_{\bm{s'}\in\mathcal{S}^{\hat{\delta}}}\left[P_{\bm{s,s'}}(u)+\sum_{\bm{r}\in\mathcal{S\setminus S}^{\hat{\delta}}}P_{\bm{s,r}}(u)\mathbf{1}(v(\bm{r})=\bm{s'})\right]\bar{C}_{\bm{s',s_{0}}}(g)\nonumber \\
\leq & \sum_{\bm{s'}\in\mathcal{S}^{\hat{\delta}}}P_{\bm{s,s'}}(u)\bar{C}_{\bm{s',s_{0}}}(g)+\sum_{\bm{r}\in\mathcal{S\setminus S}^{\hat{\delta}}}P_{\bm{s,r}}(u)\bar{C}_{\bm{r,s_{0}}}(g)\nonumber \\
= & \sum_{\bm{s'}\in\mathcal{S}}P_{\bm{s,s'}}(u)\bar{C}_{\bm{s',s_{0}}}(g).\label{eq:approximate_1}
\end{align}

Based on the above inequality, we derive that
\begin{align}
\bar{C}^{\hat{\delta}}_{\bm{s,s_{0}}}(g) & =\mathbb{E}_{g}\left[C(\bm{s},u)+\sum_{\bm{s'}\in\mathcal{S}^{\hat{\delta}}}P^{\hat{\delta}}_{\bm{s,s'}}(u)\bar{C}_{\bm{s',s_{0}}}(g)\right]\nonumber \\
 & \leq\mathbb{E}_{g}\left[C(\bm{s},u)+\sum_{\bm{s'}\in\mathcal{S}}P_{\bm{s,s'}}(u)\bar{C}_{\bm{s',s_{0}}}(g)\right]\nonumber \\
 & =\bar{C}_{\bm{s,s_{0}}}(g).
\end{align}
Therefore, condition 1 is verified to be true.

Condition 2 can be expressed as follows:
\begin{equation}
V^{(\hat{\delta})*}=\lim_{\mu\rightarrow1}(1-\mu)V^{\hat{\delta}}_{\mu}(\bm{s})\leq\lim_{\mu\rightarrow1}(1-\mu)V_{\mu}(\bm{s})=V^{*}.
\end{equation}

To verify condition 2, we first derive the following inequality similar
to \eqref{eq:approximate_1}.
\begin{align}
 & \sum_{\bm{s'}\in\mathcal{S}^{\hat{\delta}}}P^{\hat{\delta}}_{\bm{s,s'}}(u)V_{\mu}(\bm{s'})\nonumber \\
= & \sum_{\bm{s'}\in\mathcal{S}^{\hat{\delta}}}\left[P_{\bm{s,s'}}(u)+\sum_{\bm{r}\in\mathcal{S\setminus S}^{\hat{\delta}}}P_{\bm{s,r}}(u)\mathbf{1}(v(\bm{r})=\bm{s'})\right]V_{\mu}(\bm{s'})\nonumber \\
\stackrel{(a)}{\leq} & \sum_{\bm{s'}\in\mathcal{S}^{\hat{\delta}}}P_{\bm{s,s'}}(u)V_{\mu}(\bm{s'})+\sum_{\bm{r}\in\mathcal{S\setminus S}^{\hat{\delta}}}P_{\bm{s,r}}(u)V_{\mu}(\bm{r})\nonumber \\
= & \sum_{\bm{s'}\in\mathcal{S}}P_{\bm{s,s'}}(u)V_{\mu}(\bm{s'}),\label{eq:approximate_2_value}
\end{align}
where (a) is due to the monotonicity of the value function given in
Lemmas \ref{lem:Lemma1:discount_monitonicity} and \ref{lem:lemma2:concavity}.

Then, we use mathematical induction to prove that $V^{\hat{\delta}}_{\mu}(\bm{s})\leq V_{\mu}(\bm{s})$
holds. When $n=0$, $V^{\hat{\delta}}_{\mu,0}(\bm{s})\leq V_{\mu,0}(\bm{s})$
holds obviously. Suppose that $V^{\hat{\delta}}_{\mu,n}(\bm{s})\leq V_{\mu,n}(\bm{s})$,
and then we investigate whether such inequality holds at the next
iteration. Particularly, we have
\begin{align}
V^{\hat{\delta}}_{\mu,n+1}(\bm{s}) & =\min_{u\in\{0,1\}}\left\{ C(\bm{s},u)+\mu\sum_{\bm{s'}\in\mathcal{S}^{\hat{\delta}}}P^{\hat{\delta}}_{\bm{s,s'}}(u)V^{\hat{\delta}}_{\mu,n}(\bm{s'})\right\} \nonumber \\
 & \stackrel{(a)}{\leq}\min_{u\in\{0,1\}}\left\{ C(\bm{s},u)+\mu\sum_{\bm{s'}\in\mathcal{S}^{\hat{\delta}}}P^{\hat{\delta}}_{\bm{s,s'}}(u)V_{\mu,n}(\bm{s'})\right\} \nonumber \\
 & \stackrel{(b)}{\leq}\min_{u\in\{0,1\}}\left\{ C(\bm{s},u)+\mu\sum_{\bm{s'}\in\mathcal{S}}P_{\bm{s,s'}}(u)V_{\mu,n}(\bm{s'})\right\} \nonumber \\
 & =V_{\mu,n+1}(\bm{s}),
\end{align}
where (a) is due to the induction hypothesis and (b) is derived from
Eq. \eqref{eq:approximate_2_value}. With $n\rightarrow\infty,$ we
can derive that $V^{\hat{\delta}}_{\mu}(\bm{s})\leq V_{\mu}(\bm{s})$.

Since the above two conditions hold, we have that $V^{(\hat{\delta})}\rightarrow V^{*}$
when $\hat{\delta}\rightarrow\infty.$ 

\subsection{Proof for The Property of $V^{*}_{i}(W)$\label{subsec:Proof-for-ViW}}

It is easy to see that $V^{*}_{i}(W)$ is an increasing function of
$W$, since $V^{*}_{i}(W)$ is the optimal average cost, which is
given by 
\begin{equation}
V^{*}_{i}(W)=\min_{\pi}\limsup_{T\rightarrow\infty}\frac{1}{T}\mathbb{E}_{\pi}\left[\sum^{T}_{t=1}(\delta(t)+Wu(t))\mid\bm{s}\right].
\end{equation}

Then, we show that $V^{*}_{i}(W)$ is concave in $W$, i.e., for any
$W_{1}\geq0,W_{2}\geq0$,
\begin{equation}
V^{*}_{i}(aW_{1}+(1-a)W_{2})\geq aV^{*}_{i}(W_{1})+(1-a)V^{*}_{i}(W_{2}),\label{eq:concave in W}
\end{equation}
where $a\in[0,1]$. We let $\pi$ be the optimal policy under $aW_{1}+(1-a)W_{2}$
and apply it to the system with extra cost $W_{1}$. Since $\pi$
may not be optimal under extra cost $W_{1}$, we have 
\begin{align}
V^{*}_{i}(W_{1})\leq & V^{*}_{i}(aW_{1}+(1-a)W_{2})\\
 & +(1-a)(W_{1}-W_{2})\frac{\partial V^{*}_{i}(W)}{\partial W}\Big|_{W=aW_{1}+(1-a)W_{2}}.\nonumber 
\end{align}
Similarly, we have
\begin{align}
V^{*}_{i}(W_{2})\leq & V^{*}_{i}(aW_{1}+(1-a)W_{2})\\
 & +a(W_{2}-W_{1})\frac{\partial V^{*}_{i}(W)}{\partial W}\Big|_{W=aW_{1}+(1-a)W_{2}}.\nonumber 
\end{align}
 Altogether, we show that (\ref{eq:concave in W}) holds.

\section{Simulation Results}

\begin{figure}[t]
\centering

\includegraphics[width=0.9\columnwidth]{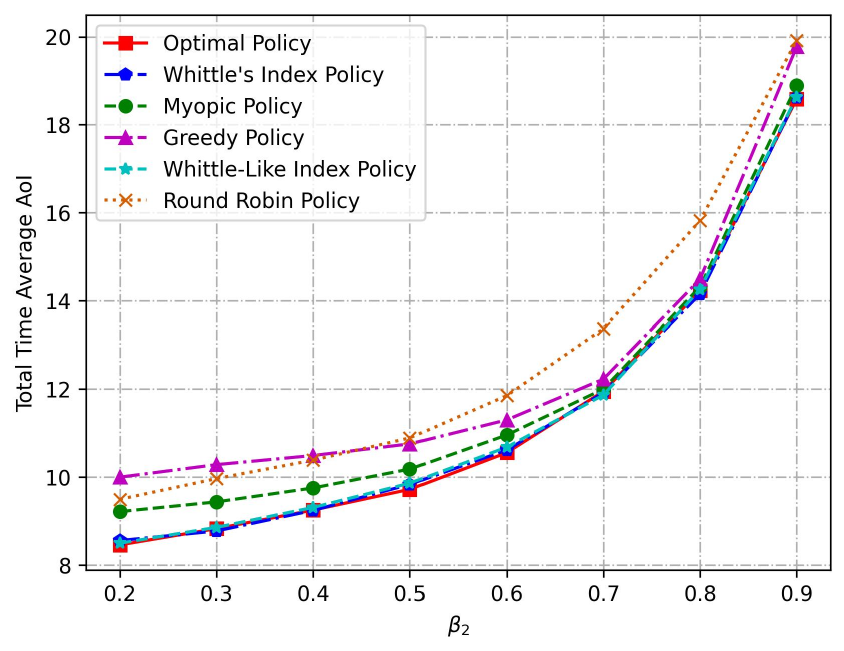}\vspace{-1em}

\caption{The total time-average AoI versus the channel state transition probability
of device 2 ($M=2$, $K=1,\alpha_{1}=0.3,\beta_{1}=0.8,$ and $\alpha_{2}=1.1-\beta_{2}$).\label{fig:beta2}}
\end{figure}
\begin{figure}[t]
\centering

\includegraphics[width=0.9\columnwidth]{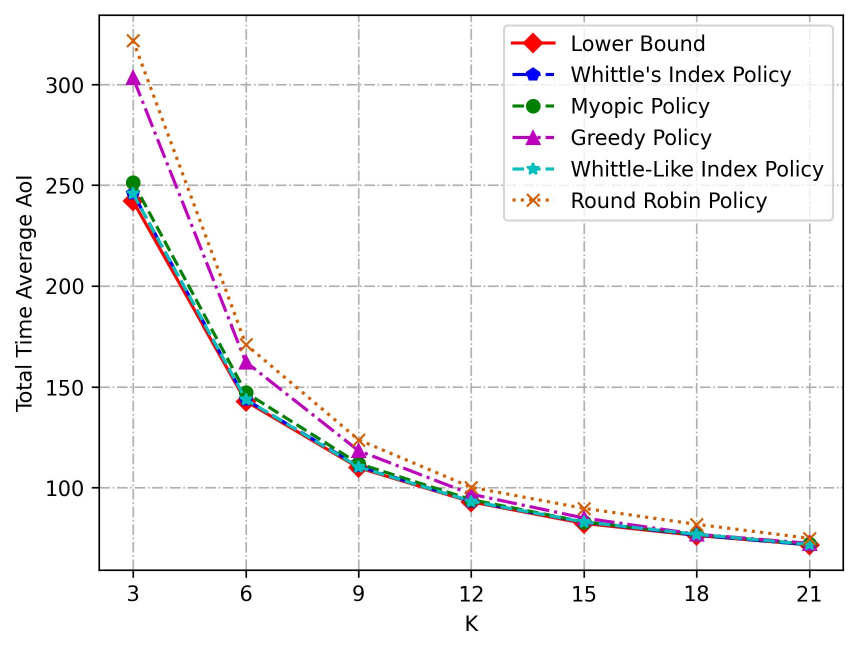}\vspace{-1em}

\caption{The total time-average AoI versus the number of scheduled devices
$K$ ($M=25$, $\alpha_{i}=\beta_{i}=0.7$).\label{fig:DifferentK}}
\end{figure}
In this section, we evaluate the performance of Whittle's index policy
and Whittle-like index policy with respect to different system parameters
and compare it with those of the greedy policy, myopic policy,{}
and Round Robin policy. Unless otherwise specified, each simulation
is run over $T=10^{5}$ time slots.  In greedy policy, the number
of $K$ devices with the largest AoI at the beginning of the slot
are scheduled, while the proposed myopic policy takes both the AoI
and the belief into consideration.{} The Round Robin
policy schedules $K$ devices cyclically in a fixed order, without
using AoI or channel-belief information. We also present the performance
of the optimal policy when the number of total devices $M$ is small
and the performance lower bound when $M$ is large.

\begin{figure}[t]
\centering

\includegraphics[width=0.9\columnwidth]{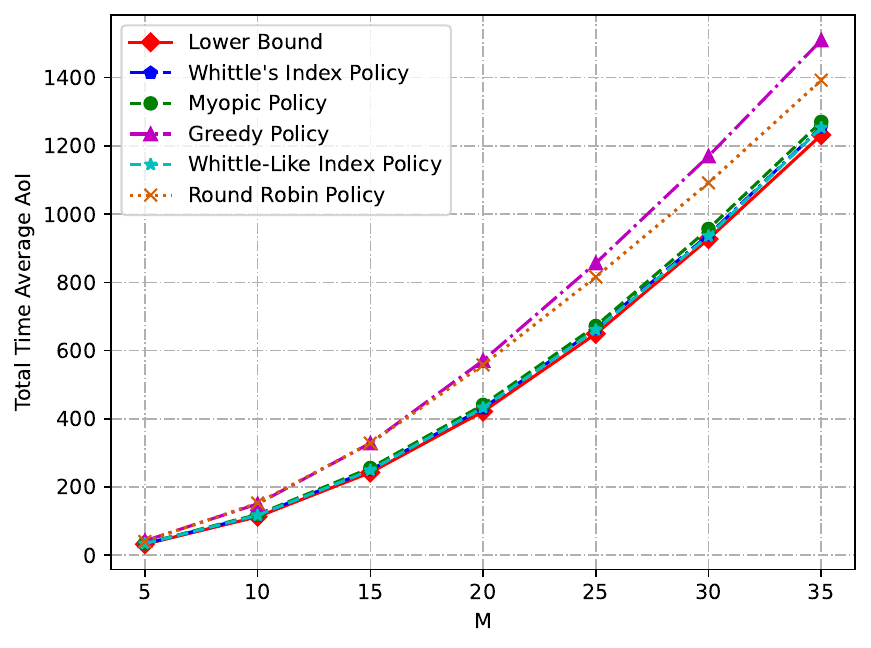}\vspace{-1em}

\caption{The total time-average AoI versus the number of total devices $M$
($K=1$, $\alpha_{i}=\beta_{i}=0.7$). \label{fig:DifferentM}}
\end{figure}

\begin{figure*}[t]
\centering\subfloat[$M=6,K=1$\label{fig:Time_a}]{\includegraphics[width=0.33\textwidth]{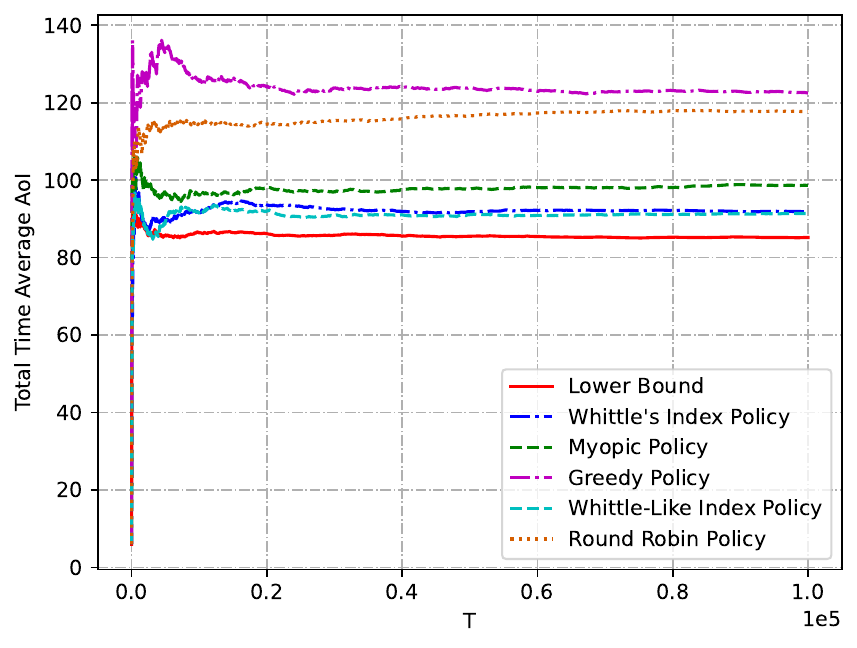}}\subfloat[$M=30,K=5$\label{fig:Time_b}]{\includegraphics[width=0.33\textwidth]{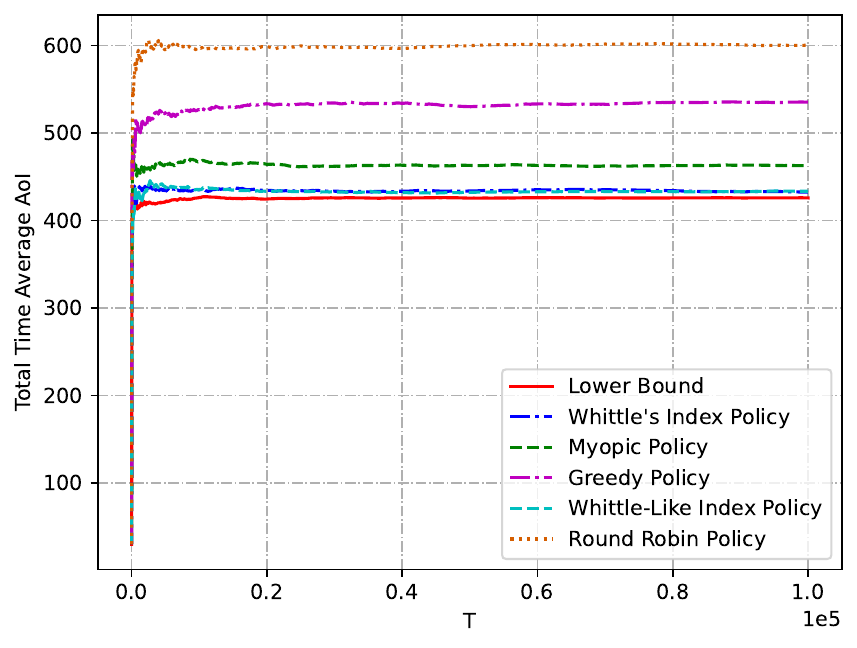}}\subfloat[$M=300,K=50$\label{fig:Time_c}]{\includegraphics[width=0.33\textwidth]{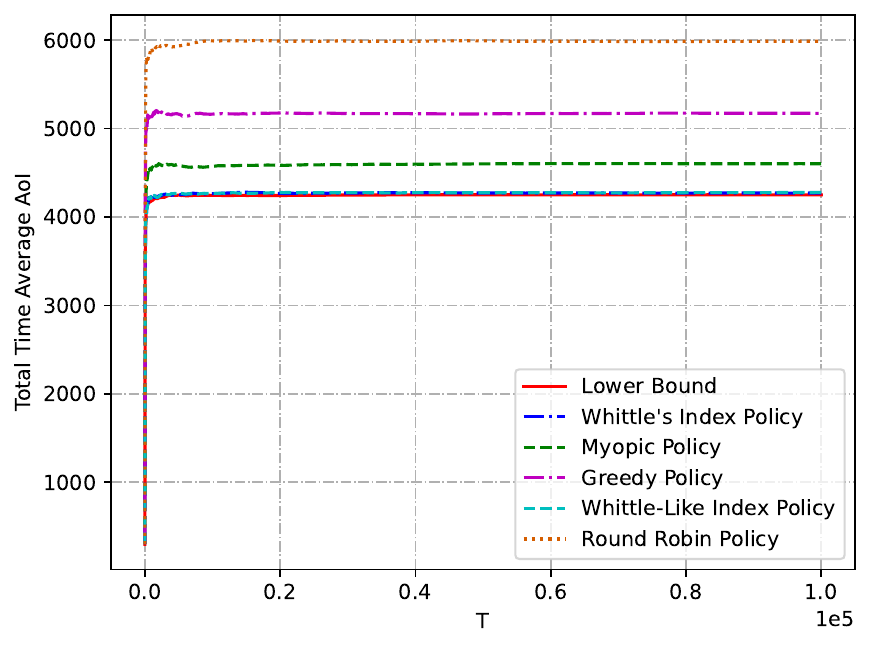}}\caption{The total time-average AoI varying with time $T$ ($\alpha_{(1)}=\beta_{(1)}=0.7,\alpha_{(2)}=0.5,\beta_{(2)}=0.9,\alpha_{(3)}=0.4,\beta_{(3)}=0.8$)\label{fig:Average_AoI_Time}}
\end{figure*}

In Fig. \ref{fig:beta2}, we show the total time-average AoI with
respect to the channel state transition probability of device 2 in
an IoT system with two devices. The channel state transition probabilities
are fixed for device 1, i.e., $\alpha_{1}=0.3,\beta_{1}=0.8$, and
the channel state transition probabilities are set as $\alpha_{2}=1.1-\beta_{2}$
for device 2 so that both channels are positively correlated. It
can be seen that the total time-average AoI increases with $\beta_{2}$
due to the growing probability of the channel being in a BAD state.
We also compare the performance of Whittle's index policy and Whittle-like
index policy with those of greedy policy, myopic policy,{}
Round Robin policy, and optimal policy. We find that the{}
Whittle's index policy and Whittle-like index policy achieve a lower
total time-average AoI than the greedy policy, myopic policy,{}
and Round Robin policy, and perform comparably with the optimal policy.{}
This result highlights the importance of exploiting channel memory
under asymmetric channel statistics. When the channel parameters of
the two devices differ, a large AoI alone does not necessarily indicate
that scheduling the corresponding device is beneficial, because the
device may have a low belief of successful transmission. The proposed
index-based policies use both AoI and channel belief, and therefore
remain close to the optimal policy across different channel conditions.
It is also shown that, when $\beta_{2}$ is approaching $0.8$ at
which the channel parameters of two devices are the same, the performance
of all policies except Round Robin policy become closer. 

In Fig. \ref{fig:DifferentK}, we illustrate the total time-average
AoI with respect to the number of scheduled devices $K$ for a fixed
$M$ and the same channel state transition probabilities. We compare
the performance of Whittle's index policy and Whittle-like index policy
against three baseline policies, along with the performance lower
bound. It is easy to see that the total time-average AoI decreases
with $K$ since the more the communication resources are, the fresher
the information at the destination. We find that Whittle's index policy,
Whittle-like index policy, and myopic policy are close to the lower
bound.{} This result shows that the benefit of intelligent
scheduling is most pronounced in the resource-scarce regime. When
only a small number of devices can be scheduled, each scheduling decision
has a larger impact on information freshness, and policies that ignore
AoI or channel belief suffer a larger performance loss. This explains
why the Round Robin policy has a larger AoI, especially when $K$
is small, because the cyclic schedule cannot prioritize devices with
large AoI or favorable channel beliefs. Moreover, as $K$ increases,
the gap between the policies and the lower bound decreases. This is
because, when almost all devices are scheduled in every time slot,
there is little difference between how the devices are selected by
each policy.

In Fig. \ref{fig:DifferentM}, we illustrate the total time-average
AoI versus the number of total devices $M$ with the same channel
state transition probabilities and compare the performance of five
policies as well as the lower bound. As at most one device can be
scheduled in each time slot (i.e., $K=1$), the transmission opportunity
for each device, on average, decreases with the increase of $M$.
Therefore, the total time-average of AoI is on the rise as the number
of total devices grows. We can also see that, as $M$ increases, the
performance of the greedy policy and Round Robin policy deteriorate
dramatically, while the performance of Whittle's index policy, Whittle-like
index policy, and myopic policy is still close to the lower bound,
which shows the outstanding performance of the proposed policies.{}
This result illustrates the scalability advantage of the proposed
index-based policies. As $M$ increases with $K$ fixed, the system
becomes increasingly resource-constrained and the cost of selecting
an unpromising device becomes larger. The greedy policy can be misled
by large AoI values when the corresponding channel belief is poor,
while Round Robin cannot adapt to either AoI or channel conditions.
In contrast, the Whittle's index and Whittle-like index policies maintain
a small gap to the lower bound, showing that the index structure remains
effective in larger networks.

In Fig. \ref{fig:Average_AoI_Time}, we show the total time-average
AoI over time for different numbers of devices $M$ with a fixed scheduling
ratio $\frac{K}{M}=\frac{1}{6}$. In particular, we divide the devices
into three groups equally, each of which has the same channel state
transition probabilities. We denote $\alpha_{(k)}$ and $\beta_{(k)}$
as the channel state transition probabilities in $k$-th group. It
is shown that our proposed policies have excellent performance in
reducing the AoI compared to the greedy policy. Moreover, the performance
of Whittle's index policy and Whittle-like index policy approach to
the lower bound with the increasing of $M$, which shows the asymptotic
optimality of Whittle's index policy.{} This result
confirms that the performance gain is persistent over time rather
than a transient effect. Across different network sizes, the Whittle's
index and Whittle-like index policies quickly stabilize near the relaxed
lower bound. The close match between these two policies also suggests
that the closed-form Whittle-like index preserves most of the scheduling
benefit of the Whittle's index while avoiding its offline computation
burden.

\section{Conclusion}

In this paper, we have studied device scheduling for information freshness
in an IoT system, where multiple devices update the status of physical
processes over time-correlated Markov channels and the scheduler does
not know the instantaneous channel state before making scheduling
decisions.{} This setting requires balancing AoI-based
freshness and belief-based transmission opportunity. We have formulated
the timely scheduling problem as a partially observable restless multi-armed
bandit problem,{} established the threshold structure
and indexability of the decoupled sub-problem, and proposed a Whittle's
index policy.We further derived a closed-form Whittle-like
index for low-complexity scheduling. The simulation results confirm
the value of combining AoI-based freshness with channel-belief information,
and show that the closed-form Whittle-like index preserves most of
the scheduling benefit with lower implementation complexity. 

Several directions remain for future work. First,
the proposed framework can be extended to heterogeneous IoT devices
with different packet sizes, energy budgets, deadline requirements,
and priority levels. Second, adaptive or learning-augmented index
policies can be developed for non-stationary IoT environments, where
channel statistics or traffic patterns may drift over time. Third,
fairness-aware AoI scheduling, weighted AoI objectives, and multi-metric
formulations that combine AoI with throughput, latency, or energy
consumption are also important extensions.

\bibliographystyle{IEEEtran}
\bibliography{AoI,Published}

\end{document}